\documentclass[draftcls, letterpaper, onecolumn]{IEEEtran}
\AtBeginDvi{}
\usepackage[cmex10]{amsmath}
\allowdisplaybreaks[1]
\usepackage{amssymb,amsthm}
\usepackage{mathrsfs}
\usepackage[noadjust]{cite}
\usepackage[dvipdfmx]{graphicx}
\usepackage{color}
\usepackage{bm}

\theoremstyle{definition}
\newtheorem{theorem}{Theorem}
\newtheorem{lemma}{Lemma}

\newtheorem{remark}{Remark}
\newtheorem{definition}{Definition}

\newtheorem{example}{Example}
\newtheorem{proposition}{Proposition}
\renewcommand{\IEEEQED}{\IEEEQEDopen}

\newcommand{\Ex}{\mathbb{E}}

\newcommand{\error}{\mathrm{P}_\mathrm{e}}
\newcommand{\SW}{\mathsf{SW}}

\newcommand{\cC}{\mathcal{C}}

\newcommand{\cE}{\mathcal{E}}

\newcommand{\cG}{\mathcal{G}}

\newcommand{\cI}{\mathcal{I}}

\newcommand{\cM}{\mathcal{M}}

\newcommand{\cP}{\mathcal{P}}
\newcommand{\cQ}{\mathcal{Q}}

\newcommand{\cS}{\mathcal{S}}

\newcommand{\cV}{\mathcal{V}}

\newcommand{\cX}{\mathcal{X}}
\newcommand{\cY}{\mathcal{Y}}
\newcommand{\cZ}{\mathcal{Z}}

\newcommand{\bmf}{\bm{f}}

\newcommand{\bms}{\bm{s}}

\newcommand{\bX}{\bm{X}}
\newcommand{\bY}{\bm{Y}}
\newcommand{\bx}{\bm{x}}
\newcommand{\by}{\bm{y}}
\newcommand{\hx}{\hat{x}}

\newcommand{\hbx}{\hat{\bm{x}}}

\newcommand{\bcX}{\overline{\cX}}

\newcommand{\markov}{ - \!\!\circ\!\! - }

\title{Distributed Computing for Functions with Certain Structures\thanks{This paper was presented in part at 2016 IEEE Information Theory Workshop at Cambridge, UK.}}
\author{Shigeaki Kuzuoka~\IEEEmembership{Member,~IEEE} and Shun Watanabe,~\IEEEmembership{Member,~IEEE}
\thanks{The work of S.~Kuzuoka is supported in part by JSPS KAKENHI Grant Number 26820145. The work of S.~Watanabe is supported
in part by JSPS KAKENHI Grant Number 16H06091.}
\thanks{S.~Kuzuoka is with the Faculty of Systems Engineering, Wakayama University, 930 Sakaedani, Wakayama, 640-8510 Japan, e-mail:kuzuoka@ieee.org.}%
\thanks{S.~Watanabe is with the Department of Computer and Information Sciences, Tokyo University of Agriculture and Technology, 2-24-16, Higashikoganeishi, Tokyo, 184-8588 Japan, e-mail:shunwata@cc.tuat.ac.jp.}%
}

\begin{document}
\flushbottom
\maketitle

\begin{abstract} 
The problem of distributed function computation is studied, where functions to be computed
is not necessarily symbol-wise. 
A new method to derive a converse bound for distributed computing is proposed; 
from the structure of functions to be computed,
information that is inevitably conveyed to the decoder is identified, and the bound is derived in terms of
the optimal rate needed to send that information. The class of informative functions is introduced, and,
for the class of smooth sources, the optimal rate for computing those functions is characterized. 
Furthermore, for i.i.d. sources with joint distribution that may not be full support,
functions that are  composition of symbol-wise function and the type of a sequence
are considered, and the optimal rate for computing those functions is characterized in terms
of the hypergraph entropy.
As a byproduct, our method also provides a conceptually simple
proof of the known fact that computing a Boolean function may require
as large rate as reproducing the entire source.  
\end{abstract}




\section{Introduction}

We study the problem of distributed computation, where the encoder observes $X^n$,
the decoder observes $Y^n$, and the function $f_n(X^n,Y^n)$ is to be computed
at the decoder based on the message sent from the encoder; see Fig.~\ref{Fig:1}. 
A straightforward scheme to compute a function is to use the Slepian-Wolf coding \cite{SlepianWolf73}.
However, since the decoder does not have to reproduce $X^n$ itself, the Slepian-Wolf rate can be improved in general. 
Then, our interest is how much improvement we can attain. 

The literature of distributed computation can be roughly categorized into two directions:\footnote{Here, we only review
papers that are directly related to this work. The problem of distributed function computation (with interactive communication)
has been actively studied in the computer science community as well \cite{Yao79,KusNis97}.}
symbol-wise functions and sensitive functions. For symbol-wise functions and the class of i.i.d. sources 
with positivity condition, i.e., i.i.d. sources for which all pairs of source symbols have positive probability, Han and Kobayashi derived the condition on functions 
such that the Slepian-Wolf rate cannot be improved at all \cite{HanKobayashi87}. 
In \cite{OrlitskyRoche01}, for i.i.d. sources that are not necessarily positive,
Orlitsky and Roche characterized the optimal rate for computing symbol-wise functions
in terms of the graph entropy introduced by K\"orner \cite{Korner73}. Particularly for i.i.d. sources with positivity condition, their result 
gives a simple characterization of the improvement of the optimal rate over the Slepian-Wolf rate. 

On the other hand, Ahlswede and Csisz\'ar introduced the class of sensitive functions,
which are not necessarily symbol-wise; they showed that, for computation of sensitive functions, the Slepian-Wolf rate cannot be
improved at all for the class of i.i.d. sources with positivity condition \cite{AhlswedeCsiszar81}.
A remarkable feature of sensitive functions is that, even though their image sizes are negligibly small compared to the input sizes,
as large rate as reproducing the entire source is needed. 
Later, a simple proof of their result was given by El Gamal \cite{ElGamal83}, and their result was 
extended by the authors to the class of smooth sources \cite{DistribComp}; the class of smooth sources includes 
sources with memory, such as Markov sources with positive transition matrices, and non-ergodic sources, such as 
mixtures of i.i.d. sources with positivity condition, which enables us to study distributed computation for a variety of 
sources in a unified manner. 

\begin{figure}[tb]
\centering{
\includegraphics[width=0.4\textwidth]{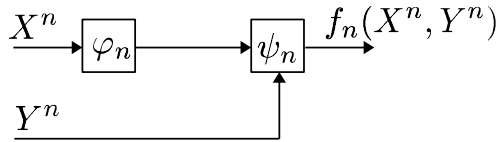}
\caption{Distributed computing with full-side-information}
\label{Fig:1}
}
\end{figure}

As described above, distributed computation for symbol-wise functions is quite well-understood; other than symbol-wise functions, 
our understanding of distributed computation is limited to the extreme case, i.e., the class of sensitive functions. Our motivation of this paper is to
further understand distributed computing for functions that are not symbol-wise nor sensitive.

\subsection{Contributions}

Our main technical contribution of this paper is a new method to derive a converse bound on distributed computation. 
A high level idea of our method is as follows: from the nature of distributed computing and the structure of the function to be computed, 
we identify information that is inevitably conveyed to the decoder. Then, we derive a bound in terms of the optimal rate needed to send that information.
As a by product, our method provides a conceptually simple proof of the above mentioned result \cite{AhlswedeCsiszar81, ElGamal83, DistribComp}
for a subclass of sensitive functions.\footnote{For instance, our method applies for the joint type, the Hamming distance, or
the inner product; see Examples \ref{example:joint_type} and \ref{example:inner_product}}.

As a class of functions such that our converse method is effective, we introduce 
the class of {\em informative functions}. 
For instance, the class includes compositions of functions where inner functions are symbol-wise
and outer functions are the type of a sequence or the modulo sum.\footnote{In a preliminary version of this paper 
published in ITW2016, these functions are investigated separately; the class of informative functions 
unifies the class of functions considered in the preliminary version.}
For the class of smooth sources,
we characterize the optimal rate for computing those functions
in terms of the Slepian-Wolf rate of an equivalence class of sources induced by the function to be computed (Theorem \ref{thm:full_support}).

Furthermore, as another application of our converse method, 
for the class of i.i.d. sources that are not necessarily positive, we characterize the optimal 
rate for computing composition of functions where inner functions are symbol-wise
and outer functions are the type of a sequence (Theorem \ref{theorem:non-full-support}). 
Even for the case of computing the joint type, which is the composition of the identity function and the type function, 
our result is novel since it is not covered by the result by Ahlswede-Csisz\'ar \cite{AhlswedeCsiszar81}.\footnote{In \cite{AhlswedeCsiszar81},
they considered a condition that is slightly weaker than the positivity condition (cf.~\cite[Theorem 2]{AhlswedeCsiszar81}); our result applies for sources that 
do not satisfy their weaker condition.} The optimal rate is characterized in terms of the hypergraph entropy, which is a natural extension of the graph entropy
of K\"orner (cf.~\cite{KornerMarton88}). In other words, our result gives an operational interpretation
to the hypergraph entropy.

Perhaps, the utility of our result can be best understood by comparing it with
the result of Orlitsky and Roche \cite{OrlitskyRoche01} via their example (cf.~Example \ref{example:comparison-main-theorem}): 
Consider an $n$-round online game, where in each round Alice and Bob each select one card without replacement from a 
virtual hat with three cards labeled 0, 1, 2. The one with larger number wins.
If Bob would like to know who won in each round, it suffices for Alice to send a message at rate $\frac{2}{3} h\left(\frac{1}{4}\right)$,
which is optimal \cite{OrlitskyRoche01}. Now, suppose that Bob does not care who won in each round; instead, he is only interested in
the total number of rounds he won. Then, our result says that it suffices for Alice to send a message at rate $\frac{1}{2}$, which is optimal.

\subsection{Organization of Paper}

The rest of the paper is organized as follows. 
In Sec.~\ref{sec:problem}, we introduce the problem formulation.
In Sec.~\ref{sec:idea}, we illustrate our converse method by using a simple example of
the inner product function. We also motivate the definition of the class of informative functions there.
In Sec.~\ref{sec:full_support}, we formally introduce the class of informative functions, and, for the class of smooth sources,
we characterize the optimal rate for computing those functions.
In Sec.~\ref{sec:non_full_support}, for the class of i.i.d. sources that are not necessarily positive,
we characterize the optimal rate for computing compositions of symbol-wise functions and the type function. 
In Sec.~\ref{sec:counclusion}, we close the paper with some conclusion and discussion. 
Some technical results and proofs of lemmas are given in  appendices.

\subsection{Notation}
Throughout this paper, random variables (e.g., $X$) and their
realizations (e.g., $x$) are denoted by capital and lower case letters
respectively.  All random variables take values in some finite alphabets
which are denoted by the respective calligraphic letters (e.g., $\cX$).
The probability distribution of random variable $X$ is denoted by $P_X$.
The support set of the distribution $P_X$ is denoted by $\mathsf{supp}(P_{X})$.
Similarly, $X^n:=(X_1,X_2,\dots,X_n)$ and
$x^n :=(x_1,x_2,\dots,x_n)$ denote, respectively, a random vector and
its realization in the $n$th Cartesian product $\cX^n$ of $\cX$.  We
will use bold lower letters to represent vectors if the length $n$ is
apparent from the context; e.g., we use $\bx$ instead of $x^n$.

For a finite set $\cS$, the cardinality of $\cS$ is denoted by $\lvert\cS\rvert$.
Given a sequence $\bms$ in the $n$th Cartesian product $\cS^n$ of $\cS$, the type $P_{\bms}=(P_{\bms}(s):s\in\cS)$ of $\bms$ is defined by
\begin{align}
 P_{\bms}(s):=\frac{\lvert\{i\in[1:n]:s_i=s\}\rvert}{n},\quad s\in\cS
\end{align}
where $[1:n]:=\{1,2,\dots,n\}$.
The set of all types of sequences in $\cS^n$ is denoted by $\cP_n(\cS)$.
The indicator function is denoted by $\bm{1}[\cdot]$.
Information-theoretic quantities are denoted in the usual manner
\cite{Cover2,CsiszarKorner2nd}.  For example, $H(X|Y)$ denotes the
conditional entropy of $X$ given $Y$.  The binary entropy is denoted by $h(\cdot)$. All logarithms are with respect
to base 2.

\section{Problem Formulation}\label{sec:problem}

Let $(\bm{X},\bm{Y}) = \{(X^n,Y^n) \}_{n=1}^\infty$ be a general correlated source with finite alphabet 
$\cX$ and $\cY$; the source is general in the sense of \cite{HanVerdu93}, i.e., it may have memory and may not be stationary nor ergodic. 
Later, we will specify a class of sources we consider in each section. 
Without loss of generality, we assume $\cX = \{0,1,\ldots,|\cX|-1\}$ and $\cY = \{0,1,\ldots,|\cY|-1\}$. 
We consider a sequence 
$\bm{f} = \{f_n\}_{n=1}^\infty$ of functions $f_n: \cX^n \times \cY^n \to \cZ_n$. A code $\Phi_n = (\varphi_n,\psi_n)$ 
for computing $f_n$ is defined by an encoder $\varphi_n:\cX^n \to \cM_n$ and a decoder $\psi_n:\cM_n \times \cY^n \to \cZ_n$.
The error probability of the code $\Phi_n$ is given by
\begin{align*}
\error(\Phi_n|f_n) := \Pr\left(  \psi_n(\varphi_n(X^n),Y^n) \neq f_n(X^n,Y^n) \right).
\end{align*}

\begin{definition}
For a given source $(\bm{X},\bm{Y})$ and a sequence of functions $\bm{f}$, a rate $R$ is defined to be
achievable if there exists a sequence $\{ \Phi_n \}_{n=1}^\infty$ of codes satisfying
\begin{align*}
\lim_{n\to \infty} \error(\Phi_n|f_n) &= 0
\intertext{and}
\limsup_{n\to\infty} \frac{1}{n} \log |\cM_n| &\le R.
\end{align*}
The optimal rate for computing $\bm{f}$, denoted by $R(\bm{X}|\bm{Y}|\bm{f})$, is the infimum of all achievable rates.
\end{definition}

\begin{definition}[SW Rate]
For a given source $(\bm{X},\bm{Y})$, the optimal rate $R(\bm{X}|\bm{Y}|\bm{f}^{\mathsf{id}})$
for the sequence $\bm{f}^{\mathsf{id}} = \{ f_n^{\mathsf{id}} \}_{n=1}^\infty$ of identity functions is 
called the {\em Slepian-Wolf (SW) rate}, and denoted by $R_{\mathsf{SW}}(\bm{X}|\bm{Y})$.
\end{definition}
Note that $R_{\mathsf{SW}}(\bm{X}|\bm{Y})$ is a trivial upper bound on $R(\bm{X}|\bm{Y}|\bm{f})$.

The following class of sources was introduced in \cite{DistribComp}, and it plays an important role
in Section \ref{sec:idea} and Section \ref{sec:full_support}.
\begin{definition}[Smooth Source]
A general source $(\bm{X},\bm{Y})$ is said to be {\em smooth} with respect to $\bm{Y}$ if
there exists a constant $0 < q < 1$, which does not depend on $n$, satisfying
\begin{align*}
P_{X^n Y^n}(\bm{x},\hat{\bm{y}}) \ge q P_{X^n Y^n}(\bm{x},\bm{y})
\end{align*}
for every $\bm{x} \in \cX^n$ and $\bm{y},\hat{\bm{y}} \in \cY^n$ with $d_H(\bm{y},\hat{\bm{y}}) =1$,
where $d_H(\cdot,\cdot)$ is the Hamming distance.
\end{definition}

The class of smooth sources is a natural generalization of i.i.d. sources with positivity 
condition studied in \cite{AhlswedeCsiszar81, HanKobayashi87}, and 
enables us to study distributed computation for a variety of sources in a unified manner. 
Indeed this class contains 
sources with memory, such as Markov sources with positive transition matrices, or
non-ergodic sources, such as mixtures of i.i.d. sources with positivity condition;
see \cite{DistribComp} for the detail.

\section{Motivating Idea}\label{sec:idea}

To explain the idea of converse proof approach used throughout the paper, let us consider the inner product function
$f_n(\bm{x},\bm{y}) = \bm{x} \cdot \bm{y}$ as a simple example, where ${\cal X}={\cal Y}=\{0,1\}$ and the inner product is computed
with modulo $2$. In fact, since the inner product function is a sensitive function in the sense of \cite{AhlswedeCsiszar81},
the optimal rate for computing the function is $R(\bm{X}|\bm{Y}|\bm{f}) = R_{\mathsf{SW}}(\bm{X}|\bm{Y})$ for every smooth
sources \cite{DistribComp}. We shall provide a conceptually simple proof of this statement in this section. 

Let $(\varphi_n,\psi_n)$ be a code with vanishing error probability:
\begin{align}
 \Pr\left(\psi_n(\varphi_n(X^n),Y^n)\neq f_n(X^n,Y^n)\right)&=\sum_{\bx,\by}P_{X^nY^n}(\bx,\by)\bm{1}\left[\psi_n(\varphi_n(\bx),\by)\neq f_n(\bx,\by)\right]\\
&\leq\varepsilon_n.
\end{align}
Since message $\varphi_n(X^n)$ is encoded without knowing the realization of side-information $Y^n$, if we input $Y^n \oplus \bm{e}_i$
to $\psi_n(\varphi_n(X^n), \cdot)$ instead of $Y^n$, we expect it will outputs $f_n(X^n,Y^n \oplus \bm{e}_i)$ with high probability, where 
$\bm{e}_i$ is a vector such that $i$th component is $1$ and other components are $0$. In fact, this intuition is true and the following bound holds:
\begin{align}
\lefteqn{ \Pr\left( \psi_n(\varphi_n(X^n), Y^n \oplus \bm{e}_i) \neq f_n(X^n, Y^n \oplus \bm{e}_i) \right) } \\
&= \sum_{\bm{x},\bm{y}} P_{X^n Y^n}(\bm{x},\bm{y}) \mathbf{1}\left[ \psi_n(\varphi_n(\bm{x}), \bm{y} \oplus \bm{e}_i) \neq f_n(\bm{x},\bm{y} \oplus \bm{e}_i) \right] \\
&\le \sum_{\bm{x},\bm{y}} \frac{1}{q} P_{X^n Y^n}(\bm{x},\bm{y} \oplus \bm{e}_i) \mathbf{1}\left[ \psi_n(\varphi_n(\bm{x}), \bm{y} \oplus \bm{e}_i) \neq f_n(\bm{x},\bm{y} \oplus \bm{e}_i) \right] \\
&\le \frac{\varepsilon_n}{q},
\end{align}
where the first inequality follows from the property of the smooth source, and the second inequality follows from the definition
of the error probability. Then, by the union bound, the above bound implies 
\begin{align}
\Pr\left( \psi_n(\varphi_n(X^n),Y^n)\neq f_n(X^n,Y^n) \mbox{ or } \psi_n(\varphi_n(X^n), Y^n \oplus \bm{e}_i) \neq f_n(X^n, Y^n \oplus \bm{e}_i)  \right) \le \frac{2 \varepsilon_n}{q}.
\end{align}
If the decoder reproduces $f_n(X^n,Y^n)$ and $f_n(X^n, Y^n \oplus \bm{e}_i)$ correctly, it can compute $X_i = X^n \cdot Y^n \oplus X^n \cdot (Y^n \oplus \bm{e}_i)$.
By conducting this procedure for $1 \le i \le n$, the decoder can reproduce an estimate $W^n$ of $X^n$ such that the bit error probability is small:
\begin{align}
\mathbb{E}\left[ \frac{1}{n} d_H(X^n, W^n) \right]
&=  \sum_{i=1}^n \frac{1}{n} \Pr( X_i \neq W_i) \\
&\le \frac{2 \varepsilon_n}{q}.
\end{align} 
By the Markov inequality, for any $\beta > 0$, we have
\begin{align}
\Pr\left( \frac{1}{n} d_H(X^n,W^n) \ge \beta \right) \le \frac{2 \varepsilon_n}{q\beta}.
\end{align}
From Lemma \ref{lemma:boosting} in Appendix \ref{appendix:misc}, if the encoder send additional message of negligible rate $\delta$, then the decoder 
can reproduce an estimate $\hat{X}^n$ of $X^n$ such that the block error probability is small:
\begin{align}
\Pr\left( X^n \neq \hat{X}^n \right) \le \frac{2 \varepsilon_n}{q\beta} + \nu_n(\beta) 2^{- n \delta}.
\end{align}
In fact, the righthand side of the above bound vanishes as $n\to\infty$ if we adjust parameters $\beta, \delta$ appropriately. 
This means that there exists a Slepian-Wolf coding scheme whose rate is asymptotically the same as the given code $(\varphi_n,\psi_n)$
to compute $f_n$, which implies $R(\bm{X}|\bm{Y}|\bm{f}) \ge R_{\mathsf{SW}}(\bm{X}|\bm{Y})$.

Two key observations of the above argument are the following: because of the nature of distributed computing,
i.e., the message is encoded without knowing the realization of side-information,
the list $(f_n(\bm{x},b\bm{y}^{(-i)}): b \in {\cal Y})$ is inevitably conveyed to the decoder with small error probability, 
where $b\bm{y}^{(-i)}$ is the sequence such that $i$th element $y_i$ of $\bm{y}$ is replaced by $b$;
and $i$th element $x_i$ can be determined from the list. 
More precisely, $x_i$ is determined by the function $\xi_n^{(i)}$ defined by
\begin{align} \label{eq:xi-inner-product}
x_i = \xi_n^{(i)}\left(\left(f_n(\bm{x},b\by^{(-i)}):b\in\cY\right)\right) :=
f_n(\bm{x},0\by^{(-i)}) \oplus f_n(\bm{x},1\by^{(-i)}).
\end{align}
Because of these facts, computing the inner product requires as large rate as the Slepian-Wolf coding.
For functions other than the inner product function, the list may not determine the value of $x_i$ in general; however, 
the list may give some partial information about $x_i$, i.e., a subset of ${\cal X}$ such that $x_i$ belongs. 
In the next section, we will introduce a class of functions that have such a property.

\section{Results for Smooth Sources}\label{sec:full_support}

\subsection{Informative Functions}

Let $\bcX$ be a partition of $\cX$; i.e., $\bcX=\{\cC_1,\cC_2,\dots,\cC_t\}$ is a set of nonempty subsets $\cC_i\subseteq\cX$ ($i=1,\dots,t$) satisfying $\cC_i\cap\cC_j=\emptyset$ ($i\neq j$) and $\cX=\bigcup_{\cC\in\bcX}\cC$.
For each $x\in\cX$, the subset $\cC\in\bcX$ satisfying $x\in\cC$ is uniquely determined and denoted by $[x]_{\bcX}$.
For a sequence $\bx\in\cX^n$, let $[\bx]_{\bcX}:=([x_1]_{\bcX},[x_2]_{\bcX},\dots,[x_n]_{\bcX})$.

For a symbol $a\in\cX$, a sequence $\bx\in\cX^n$, and an index $i\in[1:n]$, let $a\bx^{(-i)}$ be the sequence such that $i$th element $x_i$ of $\bx$ is replaced by $a$. For $b\in\cY$, $\by\in\cY^n$, and $i\in[1:n]$, $b\by^{(-i)}$ is defined similarly.
For a given permutation $\sigma$ on $[1:n]$ and a sequence $\bx\in\cX^n$, 
we denote by $\sigma(\bx)$ the sequence $\hbx\in\cX^n$ satisfying $\hx_i=x_{\sigma(i)}$ for every $i\in[1:n]$.

In the last paragraph of Section \ref{sec:idea}, for the inner product function $f_n$, 
we observed that $i$th symbol $x_i$ can be determined from the list $(f_n(\bm{x},b\bm{y}^{(-i)}): b \in {\cal Y})$
by the function $\xi_n^{(i)}$ defined by \eqref{eq:xi-inner-product}. In other words, for the finest partition ${\cal X} \equiv \{ \{0\}, \{1\}\}$,
the function $\xi_n^{(i)}$ identifies which subset of the partition $x_i$ belongs to. For functions other than the inner product,
the function identifying the subset may not exists for the finest partition; however, even in that case, a function identifying
the subset for a coarser partition may exist. Motivated by this observation, we shall introduce the class of informative
function as follows.

\begin{definition}[Informative Function]
\label{def:informative}
Let $\bcX$ be a partition of $\cX$.
A function $f_n$ is said to be $\bcX$-informative if $f_n$ satisfies the following conditions:
\begin{enumerate}
 \item\label{c1:def:informative} For each $i\in[1:n]$,
 there exists a mapping $\xi_n^{(i)}\colon\cZ_n^{\lvert\cY\rvert}\to\bcX$ such that, for any $a\in\cX$ and $(\bx,\by)\in\cX^n\times\cY^n$,
\begin{align} \label{eq:informative-1}
 \xi_n^{(i)}\left(\left(f_n(a\bx^{(-i)},b\by^{(-i)}):b\in\cY\right)\right)=[a]_{\bcX}.
\end{align}
 \item\label{c2:def:informative} For every $(\bx,\by)\in\cX^n\times\cY^n$ and any permutation $\sigma$ on $[1:n]$ satisfying $[\sigma(\bx)]_{\bcX}=[\bx]_{\bcX}$,
\begin{align}
 f_n(\sigma(\bx),\by)=f_n(\bx,\by).
\end{align}
\end{enumerate}
\end{definition}

\begin{remark}
Condition (\ref{c1:def:informative}) of Definition \ref{def:informative} can be rewritten as follows:
for each $i\in[1:n]$ and for any $a\in\cX$, the subset $\cC\in\bcX$ satisfying $a\in\cC$ can be uniquely determined from the list 
$\left(f_n(a\bx^{(-i)},b\by^{(-i)}):b\in\cY\right)$ irrespective of $(\bx,\by)\in\cX^n\times\cY^n$.
\end{remark}

Condition (\ref{c1:def:informative}) of Definition \ref{def:informative}, which will be used 
in the converse part of Theorem \ref{thm:full_support}, 
is motivated by the converse argument described in Section \ref{sec:idea}.
On the other hand, Condition (\ref{c2:def:informative}) of Definition \ref{def:informative} will be used in the achievability
part of Theorem \ref{thm:full_support}. Although the motivation of Condition (\ref{c2:def:informative}) is subtle,
the following proposition partially motivates Condition (\ref{c2:def:informative}) of Definition \ref{def:informative},
which will be proved in Appendix \ref{appendix:proof-proposition:finest-partition}.

\begin{proposition} \label{proposition:finest-partition}
For a given $\bcX$-informative function $f_n$, 
the partition $\bcX =\{\cC_1,\cC_2,\dots,\cC_t\}$ is the finest partition satisfying Condition (\ref{c1:def:informative}); in other words, 
for any partition $\bcX^\prime = \{\cC_1^\prime,\cC_2,\dots,\cC_s^\prime\}$ satisfying Condition (\ref{c1:def:informative}),
it holds that,  for every $1 \le k \le t$, 
$\cC_k \subseteq \cC_\ell^\prime$ for some $1 \le \ell \le s$.
\end{proposition}

In fact, as we will see in Example \ref{example:non-informative} after Theorem \ref{thm:full_support},
Condition (\ref{c2:def:informative}) of Definition \ref{def:informative} is stronger than $\bcX$ being 
the finest partition satisfying Condition (\ref{c1:def:informative}); there exists a function
such that the finest partition satisfying Condition (\ref{c1:def:informative}) may not satisfy Condition (\ref{c2:def:informative}).

The class of functions that are 
$\bcX$-informative for some partition $\bcX$ includes several important functions as shown in Propositions \ref{prop:simbolwise}, \ref{prop:type}, and \ref{prop:modsum} below.

At first, we consider a symbol-wise function $f_n(\bx,\by)=(f(x_1,y_1),\dots,f(x_n,y_n))$ defined from $f\colon\cX\times\cY\to\cV$.
For a function $f$ on $\cX\times\cY$, let $\bcX_f$ be the partition of $\cX$ such that two symbols $x$ and $\hx$ are in the same subset if and only if $f(x,y)=f(\hx,y)$ for all $y\in\cY$.

\begin{proposition}
\label{prop:simbolwise}
A symbol-wise function $f_n\colon\cX^n\times\cY^n\to\cV^n$ defined from $f\colon\cX\times\cY\to\cV$ is $\bcX_f$-informative.
\end{proposition}

Next, we consider a composition of functions, where 
the inner function is symbol-wise and the outer function is the type. Fix a function $f\colon\cX\times\cY\to\cV=\{0,1,\dots,m-1\}$. Then, let $f_n^{\mathsf{t}}$ be the function computing the type of the symbol-wise function $f_n$ defined from $f$; i.e.,
\begin{align}
 f_n^{\mathsf{t}}(\bx,\by):= P_{f_n(\bx,\by)},\quad (\bx,\by)\in\cX^n\times\cY^n.
\label{eq:typefunction}
\end{align}
To characterize the property of $f_n^{\mathsf{t}}$, let us introduce $\hat{f}^{\mathsf{t}}\colon\cX\times\cY\to\cV\cup\{m\}$ as
\begin{align}
 \hat{f}^{\mathsf{t}}(x,y):=
\begin{cases}
 m & \text{if }f(x,\cdot)\text{ is constant}\\
 f(x,y) & \text{otherwise.}
\end{cases}
\end{align}

\begin{proposition}
\label{prop:type}
A function $f_n^{\mathsf{t}}\colon\cX^n\times\cY^n\to\cP_n(\cV)$ defined from $f\colon\cX\times\cY\to\cV$ as \eqref{eq:typefunction} is $\bcX_{\hat{f}^{\mathsf{t}}}$-informative.
\end{proposition}

Lastly we consider the modulo-sum of function values. For a given $f\colon\cX\times\cY\to\cV=\{0,1,\dots,m-1\}$, let 
$f_n^{\oplus}$ be the function defined as
\begin{align}
 f_n^{\oplus}(\bx,\by):= \sum_{i=1}^nf(x_i,y_i)\quad (\bmod\ m),\quad (\bx,\by)\in\cX^n\times\cY^n.
\label{eq:modsum}
\end{align}
To characterize the property of $f_n^{\oplus}$, let us introduce $\hat{f}^{\oplus}$ on $\cX\times\cY$ as
\begin{align}
 \hat{f}^{\oplus}(x,y):= f(x,y+1)-f(x,y)\quad (\bmod\ m),\quad y\in\cY
\end{align}
where $f(x,\lvert\cY\rvert)=f(x,0)$.

\begin{proposition}
\label{prop:modsum}
A function $f_n^{\oplus}\colon\cX^n\times\cY^n\to\cV$ defined from $f\colon\cX\times\cY\to\cV$ as \eqref{eq:modsum} is $\bcX_{\hat{f}^{\oplus}}$-informative.
\end{proposition}

\begin{example}
Let us consider a function $f$ given in Table \ref{table:example_f}.
Then we can verify that
\begin{align}
 \bcX_f &= \{\{0\},\{1,2\},\{3\},\{4\}\},\\
 \bcX_{\hat{f}^{\mathsf{t}}} &= \{\{0,4\},\{1,2\},\{3\}\},\\
 \bcX_{\hat{f}^\oplus} &= \{\{0,4\},\{1,2,3\}\}.
\end{align}

\begin{table}[tb]
\centering{
\caption{$f\colon\cX\times\cY\to\cV=\{0,1,\dots,6\}$}\label{table:example_f}
\begin{tabular}{c|ccc}
 $x\setminus y$& 0& 1& 2\\\hline
 0 & 0 & 0 & 0 \\
 1 & 1 & 2 & 3 \\
 2 & 1 & 2 & 3 \\
 3 & 4 & 5 & 6 \\
 4 & 1 & 1 & 1
\end{tabular}
}
\end{table}
\end{example}

Proofs of Propositions \ref{prop:simbolwise}, \ref{prop:type}, and \ref{prop:modsum} are given in Appendix \ref{appendix:proof_propositions}.

\subsection{Coding Theorem}

Fix a partition $\bcX$ of $\cX$.
Then, from a pair $(X^n,Y^n)$ of RVs on $\cX^n\times\cY^n$, we can define
$([X^n]_{\bcX},Y^n)$ on 
$\bcX^n\times\cY^n$ such as
\begin{align}
 P_{[X^n]_{\bcX}Y^n}(\overline{\bx},\by):=\sum_{\bx:[\bx]_{\bcX}=\overline{\bx}}P_{X^nY^n}(\bx,\by)
\end{align}
for $(\overline{\bx},\by) \in \bcX^n\times\cY^n$.
For a given source $(\bX,\bY)$ and a partition $\bcX$ of $\cX$, let $([\bX]_{\bcX},\bY)$ be the source defined by
\begin{align}
 ([\bX]_{\bcX},\bY):=\left\{([X^n]_{\bcX},Y^n)\right\}_{n=1}^\infty.
\end{align}
Note that if $(\bm{X},\bm{Y})$ is an i.i.d.\ source then  $([\bm{X}]_{\bcX},\bm{Y})$ is also i.i.d.\ source.
Further, if $(\bm{X},\bm{Y})$ is smooth with respect to $\bm{Y}$ then 
$P_{X^n Y^n}(\bm{x},\hat{\bm{y}}) \ge q P_{X^n Y^n}(\bm{x},\bm{y})$
for all $\bm{x}$ and 
$\bm{y},\hat{\bm{y}} \in \cY^n$ with $d_H(\bm{y},\hat{\bm{y}}) =1$. 
Taking a summation over $\bm{x}\in\overline{\bx}$, we have
$P_{[X^n]_{\bcX}Y^n}(\overline{\bx},\hat{\by}) \ge q P_{[X^n]_{\bcX}Y^n}(\overline{\bx},\by)$. So, we have the following proposition.

\begin{proposition}
If $(\bm{X},\bm{Y})$ is smooth with respect to $\bm{Y}$, then $([\bm{X}]_{\bcX},\bm{Y})$ is also smooth with respect to $\bm{Y}$
(with the same constant $q$).
\end{proposition}

Now we are ready to state our coding theorem for smooth sources.

\begin{theorem}
\label{thm:full_support}
Suppose that $\bmf=\{f_n\}_{n=1}^\infty$ is $\bcX$-informative for some partition $\bcX$ of $\cX$. Then, for any smooth source $(\bX,\bY)$, we have
\begin{align} \label{eq:thm:full_support}
 R(\bX|\bY|\bmf)=R_{\SW}([\bX]_{\bcX}|\bY).
\end{align}
\end{theorem}

To illustrate Theorem \ref{thm:full_support}, let us consider several examples.

\begin{example}
\label{example:joint_type}
When $f(x,y) = (x,y)$,\footnote{Without loss of generality, we identify ${\cal X} \times {\cal Y}$ with ${\cal V} = \{0,1,\ldots,|{\cal X}||{\cal Y}|-1 \}$ in this example.}
i.e., the identity function, then $f_n^\mathsf{t}(\bm{x},\bm{y})$ is the joint type of $(\bm{x},\bm{y})$. 
In this case,  $\hat{f}^\mathsf{t}$ is the identity function and $\overline{{\cal X}}_{\hat{f}^\mathsf{t}} = \{ \{ x \} : x \in {\cal X} \}$.
Thus, Proposition \ref{prop:type} and Theorem \ref{thm:full_support} imply $R(\bX|\bY|\bmf^\mathsf{t}) = R_{\mathsf{SW}}(\bX|\bY)$.
\end{example}

\begin{example}
When $f(x,y) = x$, then $f_n^\mathsf{t}(\bm{x},\bm{y})$ is the marginal type of $\bm{x}$. 
In this case, $\hat{f}^\mathsf{t}(x,y) = m$ for every $(x,y)$, i.e., the constant function, and $\overline{{\cal X}}_{\hat{f}^\mathsf{t}} = \{ {\cal X} \}$. 
Thus, Proposition \ref{prop:type} and Theorem \ref{thm:full_support} imply $R(\bX|\bY|\bmf^\mathsf{t}) = 0$.
\end{example}

\begin{example}
When ${\cal X} = {\cal Y} = \{0,1\}$ and $f(x,y) = x \oplus y$, let us consider the modulo-sum function $f_n^\oplus$ induced by $f$.
In this case, $\hat{f}^\oplus(x,y) = 1$ for every $(x,y)$ and $\overline{{\cal X}}_{\hat{f}^\oplus} = \{ {\cal X}\}$.
Thus, Proposition \ref{prop:modsum} and Theorem \ref{thm:full_support} imply $R(\bX|\bY|\bmf^\oplus) = 0$.
In fact, the encoder can just send the parity $\oplus_{i=1}^n X_i$. Then, the decoder can reproduce
$f_n^\oplus(X^n,Y^n) = (\oplus_{i=1}^n X_i) \oplus (\oplus_{i=1}^n Y_i)$.
It is interesting to compare this example with the fact that, for the same function $f(x,y) = x \oplus y$,
$R(\bX|\bY|\bmf^\mathsf{t}) = R_{\mathsf{SW}}(\bX|\bY)$.
\end{example}

\begin{example}
\label{example:inner_product}
When ${\cal X} = {\cal Y} = \{0,1\}$ and $f(x,y) = x \wedge y$, let us consider the modulo-sum function $f_n^\oplus$ induced by $f$.
In this case, $\hat{f}^\oplus(x,y) = x$ and $\overline{{\cal X}}_{\hat{f}^\oplus} = \{ \{x \} : x \in {\cal X}\}$. 
Thus, Proposition \ref{prop:modsum} and Theorem \ref{thm:full_support} imply $R(\bX|\bY|\bmf^\oplus) = R_{\mathsf{SW}}(\bX|\bY)$.
Note that $f_n^\oplus$ is the inner product function, and it recovers the result explained in Sec.~\ref{sec:idea}.
\end{example}

Finally, in order to illustrate the role of Condition (\ref{c2:def:informative}) in Definition \ref{def:informative},
let us consider a function that is not $\bcX$-informative 
for any partition $\bcX$ of ${\cal X}$, but the optimal rate can be characterized. 

\begin{example} \label{example:non-informative}
For ${\cal X} = {\cal Y} = \{0,1\}$, let $f_n: {\cal X}^n \times {\cal Y}^n \to \{0,1\}^n$ be the function defined by
\begin{align} \label{eq:non-informative-example}
f_n(\bm{x},\bm{y}) = \big( \bm{1}[x_i \oplus y_i = x_{i+1} \oplus y_{i+1}] : 1 \le i \le n  \big),
\end{align}
where $x_{n+1} = x_1$ and $y_{n+1}=y_1$. For this function, we can verify that 
the trivial partition $\bcX = \{ \{0,1\} \}$ is the only partition that satisfies Condition (\ref{c1:def:informative}) of Definition \ref{def:informative}.
In fact, if Condition (\ref{c1:def:informative}) is satisfied for partition ${\cal X} \equiv \{ \{0\}, \{1\}\}$, then,
for $\bar{\bm{x}} = (x_1\oplus 1,\ldots,x_n \oplus 1)$ and $\bar{a} = a \oplus 1$, we have
\begin{align}
\{a\} &= \xi_n^{(i)}\left(\left(f_n(a\bx^{(-i)},b\by^{(-i)}):b\in\cY\right)\right) \\
&= \xi_n^{(i)}\left(\left(f_n(\bar{a} \bar{\bx}^{(-i)},b\by^{(-i)}):b\in\cY\right)\right) \\
&= \{\bar{a}\},
\end{align}
which is a contradiction. On the other hand, Condition (\ref{c2:def:informative}) is apparently not 
satisfied for the trivial partition $\bcX = \{ \{0,1\} \}$. Thus, this function is not $\bcX$-informative 
for any partition $\bcX$ of ${\cal X}$. However, we can verify that $R(\bX|\bY|\bmf)=R_{\SW}(\bX|\bY)$ as follows.
Suppose that we are given a code to compute the function $f_n$ with
vanishing error probability. If the encoder additionally send one bit, say $x_1$, then the decoder can sequentially 
reproduce all $x_i$s from $f_n(\bm{x},\bm{y})$ and $\bm{y}$, which implies $R_{\SW}(\bX|\bY) \le R(\bX|\bY|\bmf)$.
\end{example}

As we can find from Example \ref{example:non-informative}, partition $\bcX$ being
the finest partition satisfying Condition (\ref{c1:def:informative}) does not imply \eqref{eq:thm:full_support}.
In order to handle functions as in Example \ref{example:non-informative}, we need to consider more 
general ``informative" structure of given functions, which is beyond the scope of this paper.

\subsection{Proof of Theorem \ref{thm:full_support}}
We first prove the converse part
\begin{align}
 R(\bX|\bY|\bmf)\geq R_{\SW}([\bX]_{\bcX}|\bY)
\label{eq_converse:proof:thm:full_support}
\end{align}
and then prove the direct part
\begin{align}
 R(\bX|\bY|\bmf)\leq R_{\SW}([\bX]_{\bcX}|\bY).
\label{eq_direct:proof:thm:full_support}
\end{align}

\begin{IEEEproof}
[Converse part]
Let $(\varphi_n,\psi_n)$ be a code satisfying
\begin{align}
 \Pr\left(\psi_n(\varphi_n(X^n),Y^n)\neq f_n(X^n,Y^n)\right)&=\sum_{\bx,\by}P_{X^nY^n}(\bx,\by)\bm{1}\left[\psi_n(\varphi_n(\bx),\by)\neq f_n(\bx,\by)\right]\\
&\leq\varepsilon_n.
\end{align}
Further, let $\pi_i\colon\cY^n\to\cY^n$ be the permutation that shifts only $i$th symbol of $\by\in\cY^n$; i.e., $y_i\mapsto y_i+1$ ($\bmod \lvert\cY\rvert$). Then, since $(\bX,\bY)$ is smooth, for every $\hat{b}$ ($0\leq\hat{b}\leq\lvert\cY\rvert-1$)
and $i\in[1:n]$, we have
\begin{align}
\lefteqn{\Pr\left(\psi_n(\varphi_n(X^n),\pi_i^{\hat{b}}(Y^n))\neq f_n(X^n,\pi_i^{\hat{b}}(Y^n))\right)}\nonumber\\
&=\sum_{\bx,\by}P_{X^nY^n}(\bx,\by)\bm{1}\left[\psi_n(\varphi_n(\bx),\pi_i^{\hat{b}}(\by))\neq f_n(\bx,\pi_i^{\hat{b}}(\by))\right]\\
&\leq \sum_{\bx,\by}\frac{1}{q}P_{X^nY^n}(\bx,\pi_i^{\hat{b}}(\by))\bm{1}\left[\psi_n(\varphi_n(\bx),\pi_i^{\hat{b}}(\by))\neq f_n(\bx,\pi_i^{\hat{b}}(\by))\right]\\
&\leq \frac{\varepsilon_n}{q}.\label{eq:proof:thm:full_support}
\end{align}
Thus, we have,\footnote{Here we denote by $bY^{(-i)}$ a sequence such that $Y_i$ of $Y^n$ is replaced by $b$.}
\begin{align}
\lefteqn{
 \Pr\left(
\left(\psi_n(\varphi_n(X^n),bY^{(-i)}):b\in\cY\right)\neq \left(f_n(X^n,bY^{(-i)}):b\in\cY\right)
\right)}\nonumber\\
&= \Pr\left(\exists b\in\cY\text{ s.t. }\psi_n(\varphi_n(X^n),bY^{(-i)})\neq f_n(X^n,bY^{(-i)})\right)\\
&= \sum_{\bx,\by}P_{X^nY^n}(\bx,\by)\bm{1}\left[\exists b\in\cY\text{ s.t. }\psi_n(\varphi_n(\bx),b\by^{(-i)})\neq f_n(\bx,b\by^{(-i)})\right]\\
&= \sum_{\bx,\by}P_{X^nY^n}(\bx,\by)\bm{1}\left[\exists 0\leq \hat{b}\leq\lvert\cY\rvert-1\text{ s.t. }\psi_n(\varphi_n(\bx),\pi_i^{\hat{b}}(\by))\neq f_n(\bx,\pi_i^{\hat{b}}(\by))\right]\\
&= \Pr\left(\exists 0\leq \hat{b}\leq\lvert\cY\rvert-1 \text{ s.t. }\psi_n(\varphi_n(X^n),\pi_i^{\hat{b}}(Y^n))\neq f_n(X^n,\pi_i^{\hat{b}}(Y^n))\right)\\
&\leq \frac{\lvert\cY\rvert}{q}\varepsilon_n,
\end{align}
where the last inequality follows from \eqref{eq:proof:thm:full_support} and the union bound.

Since $f_n$ is $\bcX$-informative (and thus the condition (\ref{c1:def:informative}) of Definition \ref{def:informative} holds), there exists a mapping $\xi_n^{(i)}$ such that
\begin{align}
  \Pr\left(
\xi_n^{(i)}\left(
\left(
\psi_n(\varphi_n(X^n),bY^{(-i)}):b\in\cY\right)\right)\neq [X_i]_{\bcX}
\right)
&\leq \frac{\lvert\cY\rvert}{q}\varepsilon_n.
\end{align}
Hence, we can construct a decoder $\tilde{\psi}_n\colon\cM_n\times\cY^n\to\bcX^n$ such that $W^n=\tilde{\psi}_n(\varphi(X^n),Y^n)$ satisfies
\begin{align}
\Ex\left[
\frac{1}{n}d_H([X^n]_{\bcX},W^n)
\right] &=\sum_{i=1}^n\frac{1}{n}\Pr\left([X_i]_{\bcX}\neq W_i\right)\\
&\leq \frac{\lvert\cY\rvert}{q}\varepsilon_n.
\end{align}
By the Markov inequality, for any $\beta>0$, we have
\begin{align}
 \Pr\left(\frac{1}{n}d_H([X^n]_{\bcX},W^n)\geq\beta\right)
&\leq \frac{\lvert\cY\rvert}{q\beta}\varepsilon_n.
\end{align}
Thus, by Lemma \ref{lemma:boosting} in Appendix \ref{appendix:misc}, there exists a code $(\kappa_n,\tau_n)$ of size $2^{n\delta}$ such that
\begin{align}
 \Pr\left(\tau_n(\kappa_n([X^n]_{\bcX}),W^n)\neq [X^n]_{\bcX}\right)
&\leq \frac{\lvert\cY\rvert}{q\beta}\varepsilon_n+\nu_n(\beta)2^{-n\delta}.
\end{align}
Since $[X^n]_{\bcX}$ is a function of $X^n$ and the total code size of $(\varphi_n,\tilde{\psi}_n)$ and $(\kappa_n,\tau_n)$ is $\lvert\cM_n\rvert 2^{n\delta}$, by 
Lemma \ref{lemma:converse-function-X} in Appendix \ref{appendix:misc}, we have
\begin{align}
 \Pr\left(
\frac{1}{n}\log\frac{1}{P_{[X^n]_{\bcX}|Y^n}([X^n]_{\bcX}|Y^n)}
>\frac{1}{n}\log\lvert\cM_n\rvert +2\delta
\right)
&\leq \frac{\lvert\cY\rvert}{q\beta}\varepsilon_n+(\nu_n(\beta)+1)2^{-n\delta}.
\end{align}
Thus, by the standard argument on the Slepian-Wolf coding (see, e.g.~\cite{Han-spectrum}), there exists a code for $([X^n]_{\bcX},Y^n)$ with rate $(1/n)\log\lvert\cM_n\rvert+3\delta$ such that the error probability is less than
\begin{align}
 \frac{\lvert\cY\rvert}{q\beta}\varepsilon_n+(\nu_n(\beta)+2)2^{-n\delta}.
\end{align}
By taking $\delta>0$ appropriately compared to $\beta>0$, the error probability converges to $0$, which implies
\begin{align}
R_{\SW}([\bX]_{\bcX}|\bY)\leq   R(\bX|\bY|\bmf)+3\delta.
\end{align}
Since $\beta>0$ can be arbitrarily small, and we can make $\delta>0$ arbitrarily small accordingly, we have
\eqref{eq_converse:proof:thm:full_support}.
\end{IEEEproof}

\begin{IEEEproof}
[Direct part]
First we claim that, given $[\bx]_{\bcX}$ and $P_{\bx}\in\cP_n(\cX)$ of a sequence $\bx\in\cX^n$, 
we can construct a sequence $\hbx\in\cX^n$ satisfying $[\hbx]_{\bcX}=[\bx]_{\bcX}$ and $\hbx=\sigma(\bx)$ for a permutation $\sigma$.
Indeed, we can construct $\hbx=\hbx([\bx]_{\bcX},P_{\bx})$ as follows.
From $[\bx]_{\bcX}$, we can determine a partition $\{\cI_\cC:\cC\in\bcX\}$ of $[1:n]$ as
\begin{align}
\cI_\cC:=\{i\in[1:n]:[x_i]_{\bcX}=\cC\},\quad\cC\in\bcX.
\end{align}
Then, given $P_{\bx}$, we can divide each $\cI_\cC$ ($\cC\in\bcX$) into a partition $\{\cI_a:a\in\cC\}$ so that\footnote{Although there are several partitions which satisfy \eqref{eq2_direct:proof:thm:full_support}, the choice of a partition does not affect the argument; we may choose a partition so that, for $a,\hat{a}\in\cC$ satisfying $a<\hat{a}$, if $i\in\cI_a$ and $j\in\cI_{\hat{a}}$ then $i<j$.}
\begin{align}
 \lvert\cI_a\rvert =nP_{\bx}(a),\quad a\in\cC\subseteq\cX.
\label{eq2_direct:proof:thm:full_support}
\end{align}
Note that $\{\cI_a:a\in\cX\}$ is also a partition of $[1:n]$; i.e., for each $i\in[1:n]$ there exists only one $\hx_i\in\cX$ such that $i\in\cI_{\hx_i}$. Then, it is not hard to see that $\hbx=(\hx_1,\dots,\hx_n)$ satisfies the desired property.

Now, suppose that we are given a Slepian-Wolf code $(\hat{\varphi}_n,\hat{\psi}_n)$ for sending $[X^n]_{\bcX}$ with error probability $\varepsilon_n$.
From this code, we can construct a code for computing $f_n$ as follows. In the new code, observing $X^n=\bx$, the encoder sends 
the marginal type $P_{\bx}$ of $\bx$ by using $\lvert\cX\rvert\log(n+1)$ bits in addition to 
the codeword $\hat{\varphi}_n([\bx]_{\bcX})$ of the original SW code.
Assume that the decoder can 
obtain $[\bx]_{\bcX}$ from $\hat{\varphi}_n([\bx]_{\bcX})$ and $\by$. Then, since $P_{\bx}$ is sent from the encoder, the decoder can construct a sequence $\hbx=\hbx([\bx]_{\bcX},P_{\bx})$ satisfying $[\hbx]_{\bcX}=[\bx]_{\bcX}$ and $\hbx=\sigma(\bx)$ for a permutation $\sigma$ as shown above.
Note that $\hbx$ satisfies $f_n(\hbx,\by)=f_n(\bx,\by)$, since $f_n$ is $\bcX$-informative (and thus the condition (\ref{c2:def:informative}) of Definition \ref{def:informative} holds).
This proves that the decoder can compute $f_n(\bx,\by)$ with error probability $\varepsilon_n$, and thus we have \eqref{eq_direct:proof:thm:full_support}.
\end{IEEEproof}

\section{Results for Restricted Supports}\label{sec:non_full_support}

In this section, we consider i.i.d. source $(\bm{X}, \bm{Y}) = \{(X^n,Y^n) \}_{n=1}^\infty$ distributed according 
to $P_{XY}$, where $P_{XY}$  may not be full support and the support set is denoted by ${\cal S}\subseteq {\cal X} \times {\cal Y}$. 
When $P_{XY}$ is not full support, $(\bm{X}, \bm{Y})$ is not smooth anymore. However, for the type of symbol-wise functions,
we can derive an explicit formula for $R(\bm{X}|\bm{Y}|\bm{f}^{\mathsf{t}})$ by modifying the idea explained in Section \ref{sec:idea}.

\subsection{Hypergraph Graph Entropy}

In this section, we introduce the hypergraph entropy, which is a natural extension of the
graph entropy introduced in \cite{Korner73} (see also \cite{OrlitskyRoche01} and \cite{KornerMarton88}).\footnote{More 
precisely, the quantity defined by \eqref{eq:definition-hypergraph-entropy} is called ``hyperclub entropy" in \cite{KornerMarton88} (see also \cite{CsiKorLovMarSim90}) and 
the terminology ``hypergraph entropy" is used for a different quantity in \cite{KornerMarton88}; since ``hyperclub"
is not very common terminology, we call the quantity defined by \eqref{eq:definition-hypergraph-entropy}, ``hypergraph entropy" in this paper.} 
A hypergraph ${\cal G} = ({\cal X}, {\cal E})$ consists of vertex set ${\cal X}$ 
and hyperedge set ${\cal E} \subseteq 2^{{\cal X}}$. For the purpose of this paper, we identify
the vertex set with the alphabet of the encoder's observation. 

Let $X$ be a random variable on ${\cal X}$. Without loss of generality, we assume $\mathsf{supp}(P_X) = {\cal X}$.
The hypergraph entropy of $(X,{\cal G})$ is defined by
\begin{align} \label{eq:definition-hypergraph-entropy}
H_{\cal G}(X) &:= \min_{X \in W \in {\cal E}} I(W \wedge X) \\
&= \min\bigg\{ I(W \wedge X) : W \mbox{ takes values in } {\cal E} \mbox{ and } \Pr( X \in W) = 1 \bigg\}.
\end{align}
More precisely, the minimization is taken over test channel $P_{W|X}$ satisfying
$\sum_{w \ni x} P_{W|X}(w|x) = 1$. By the data processing inequality, the minimization can be restricted to
$W$s ranging over maximal hyperedge of ${\cal E}$.

For a given (standard) graph ${\cal G}^\prime = ({\cal X}, {\cal E}^\prime)$ with edge set ${\cal E}^\prime \subseteq {\cal X} \times {\cal X}$, 
a set of vertices is independent if no two are connected by any edge in ${\cal E}^\prime$. If we choose the hyperedge set
${\cal E}$ as the set of all independent set of ${\cal G}^\prime$, then the hypergraph entropy $H_{\cal G}(X)$
is nothing but the graph entropy of $(X,{\cal G}^\prime)$ in the sense of \cite{Korner73}. 

\begin{example}[\cite{OrlitskyRoche01}] \label{example:OrlRoc-unconditional}
Let ${\cal X} = \{0,1,2\}$, ${\cal E} = \{ \{0,1\}, \{1,2\} \}$, and $P_X$ be the uniform 
distribution on ${\cal X}$.  By convexity of mutual information,
$I(W \wedge X)$ is minimized when $P_{W|X}(\{0,1\} | 1) = P_{W|X}(\{1,2\} | 1) = \frac{1}{2}$. Thus, we have
\begin{align}
H_{\cal G}(X) = H(W) - H(W|X) = 1 - \frac{1}{3} = \frac{2}{3}.
\end{align}
\end{example}

\begin{example} \label{example:modified-OrlRoc-unconditional}
Let ${\cal X} = \{0,1,2\}$, ${\cal E} = \{ \{0,1\}, \{1,2\}, \{0,2\} \}$, and $P_X$ be the uniform 
distribution on ${\cal X}$. By convexity of mutual information and symmetry,\footnote{In fact, by rotating the labels
$0,1,2$ and by using convexity, we can first show that $P_{W|X}(\{0,1\}|0)=P_{W|X}(\{1,2\}|1) = P_{W|X}(\{0,2\}|2) = \alpha$
and $P_{W|X}(\{0,2\}|0) = P_{W|X}(\{0,1\}|1) = P_{W|X}(\{1,2\}|2) = \beta$ for some
$\alpha,\beta$ with $\alpha+\beta=1$ is optimal. Then, by flipping $P_{W|X}(w|x), P_{W|X}(w^\prime|x)$ for
each $w,w^\prime \ni x$ and by using convexity, we can show that $\alpha=\beta$ is optimal.}
$I(W \wedge X)$ is minimized when $P_{W|X}(w|x) = \frac{1}{2}$ for every $w \ni x$. Thus, we have
\begin{align}
H_{\cal G}(X) = H(W) - H(W|X) = \log 3 - 1.
\end{align}
\end{example}

Next, let us extend the above definition to the conditional hypergraph entropy.
Let $(X,Y)$ be a pair of random variables on ${\cal X} \times {\cal Y}$. The hypergraph
entropy of $(X,{\cal G})$ given $Y$ is defined by 
\begin{align}
H_{\cal G}(X|Y) := \min_{W \markov X \markov Y \atop X \in W \in {\cal E}} I(W \wedge X|Y),
\end{align}
where $W \markov X \markov Y$ indicates that $W, X, Y$ form a Markov chain.

\begin{example}[\cite{OrlitskyRoche01}] 
For the same hypergraph as Example \ref{example:OrlRoc-unconditional}, ${\cal Y} = {\cal X}$, 
${\cal S} = \{ (x,y) : x,y \in \{0,1,2\},  x \neq y \}$, and $P_{XY}(x,y) = \frac{1}{6}$ for $(x,y) \in {\cal S}$, by the convexity of the conditional mutual information, 
$I(W \wedge X|Y)$ is minimized  when $P_{W|X}(\{0,1\} | 1) = P_{W|X}(\{1,2\} | 1) = \frac{1}{2}$. Thus, we have
\begin{align}
H_{\cal G}(X|Y) &= H(W|Y) - H(W|X,Y) \\
&= \frac{1}{3} + \frac{2}{3} h\left( \frac{1}{4} \right) - \frac{1}{3} \\
&= \frac{2}{3} h\left( \frac{1}{4} \right).
\end{align}
\end{example}

\begin{example} \label{example:conditional-hypergraph-entropy}
For the same hypergraph as Example \ref{example:modified-OrlRoc-unconditional}, ${\cal Y} = {\cal X}$, 
${\cal S} = \{ (x,y) : x,y \in \{0,1,2\},  x \neq y \}$, and $P_{XY}(x,y) = \frac{1}{6}$ for $(x,y) \in {\cal S}$, by the convexity of the conditional mutual information
and symmetry, 
$I(W \wedge X|Y)$ is minimized  when $P_{W|X}(w|x) = \frac{1}{2}$ for every $w \ni x$. Thus, we have
\begin{align}
H_{\cal G}(X|Y) &= H(W|Y) - H(W|X,Y) \\
&= \frac{3}{2} - 1 \\
&= \frac{1}{2}.
\end{align}
\end{example}

\subsection{Compatible Hyperedge and Solvable Hyperedge}

In this section, we introduce concepts of {\em compatible hyperedge} and {\em solvable hyperedge}.
These concepts play important roles in the statement as well as the proof of Theorem \ref{theorem:non-full-support} in latter sections.
To facilitate understanding of the concepts, we will provide some examples to support definitions.
All lemmas given in this section will be proved in Appendix \ref{appendix:proof_lemmas}.
To simplify the notation, let us introduce
\begin{align}
 \cQ_n := \cP_n(\cV).
\end{align}

When we considered smooth sources in the previous section, the list $(f_n^\mathsf{t}(\bm{x}, b \bm{y}^{(-i)}) : b \in {\cal Y}) \in {\cal Q}_n^{|{\cal Y}|}$ played an important role.
On the other hand, since the source we consider in this section does not have full support, the component function $f(x,y)$ is undefined 
on the complement ${\cal S}^c$
of the support set ${\cal S}$.\footnote{Even if it is defined, there is no guarantee that a given code recovers the correct value of $f_n^\mathsf{t}(\bm{x},\bm{y})$
when $(x_i,y_i) \notin {\cal S}$ for some $i \in [1:n]$.} Thus, we need to consider the set of possible lists for given $(\bm{x},\bm{y})$ and $i \in [1:n]$
by assuming values on ${\cal S}^c$ are appended arbitrarily, which is defined as follows. 

\begin{definition}
For a given $(\bm{x},\bm{y}) \in {\cal S}^n$ and for each $i \in [1:n]$, let ${\cal Q}_n^{(i)}({\cal S},\bm{x},\bm{y})$ be the set of all 
$(Q_0,\ldots,Q_{|{\cal Y}|-1}) \in {\cal Q}_n^{|{\cal Y}|}$ satisfying 
\begin{align} \label{eq:correct-sequence}
Q_b = f_n^{\mathsf{t}}(\bm{x}, b \bm{y}^{(-i)}),~~~\forall b \mbox{ s.t. } (x_i, b) \in {\cal S}.
\end{align}
\end{definition}

In the converse proof of Theorem \ref{theorem:non-full-support}, for a given list $(Q_0,\ldots,Q_{|{\cal Y}|-1}) \in {\cal Q}_n^{(i)}({\cal S},\bm{x},\bm{y})$,
we need to infer possible values of $x_i$. The following definition provides the set of possible candidates, and Lemma \ref{lemma:compatible}
guarantees that $x_i$ is included in the candidate set.\footnote{It may be worth to note that $nQ_{b}(v)$ changes by at most 1 if $b$ is changed, since $Q_b$ is the type of a symbol-wise function.}
 
\begin{definition}[Compatible Hyperedge]
For a given $(Q_0,\ldots,Q_{|{\cal Y}|-1}) \in {\cal Q}_n^{|{\cal Y}|}$, we say that $a \in {\cal X}$ is {\em compatible} with
$(Q_0,\ldots,Q_{|{\cal Y}|-1})$ if, for every $b_1,b_2 \in {\cal Y}$ with $(a,b_1) \in {\cal S}$ and $(a,b_2) \in {\cal S}$,
\begin{align}
n Q_{b_1}(v) - n Q_{b_2}(v) = \mathbf{1}[ f(a,b_1) = v] - \mathbf{1}[ f(a,b_2) = v],~~~\forall v \in {\cal V}.
\end{align}
Then, let
\begin{align}
e(Q_0,\ldots,Q_{|{\cal Y}|-1}) := \{ a : a \mbox{ is compatible with } (Q_0,\ldots,Q_{|{\cal Y}|-1}) \}
\end{align}
be the hyperedge that is compatible with $(Q_0,\ldots,Q_{|{\cal Y}|-1})$.\footnote{The set $e(Q_0,\ldots,Q_{|{\cal Y}|-1})$
can be empty set; even if it is empty, we call it a hyperedge though the empty set is commonly not regarded as a hyperedge.}
\end{definition}

\begin{lemma}
\label{lemma:compatible}
If $(Q_0,\ldots,Q_{|{\cal Y}|-1}) \in {\cal Q}_n^{(i)}({\cal S},\bm{x},\bm{y})$, then $x_i \in e(Q_0,\ldots,Q_{|{\cal Y}|-1})$.
\end{lemma}

\begin{example}
Let ${\cal Y} = {\cal X} = \{0,1,2\}$, ${\cal S} = \{ (x,y) : x,y \in \{0,1,2\},  x \neq y \}$, and 
\begin{align} \label{eq:function-card-problem}
f(x,y) = \left\{
\begin{array}{ll}
0 & \mbox{if } x > y \\
1 & \mbox{if } x < y
\end{array}
\right..
\end{align}
For $n = 6$, let 
\begin{align}
\bm{x} &= (0~1~2~1~2~0), \\
\bm{y} &= (1~0~0~2~1~2),
\end{align}
which means
\begin{align}
f_n(\bm{x},\bm{y}) &= (1~0~0~1~0~1), \\
f_n^\mathsf{t}(\bm{x},\bm{y}) &= \left( \frac{1}{2}, \frac{1}{2} \right).
\end{align}
For $i=4$, since $(x_i,1) \notin {\cal S}$, $Q_1$ can be arbitrary. Thus,
\begin{align}
{\cal Q}_n^{(4)}({\cal S}, \bm{x},\bm{y}) = \left\{ \left( \left( \frac{2}{3}, \frac{1}{3} \right), Q_1, \left( \frac{1}{2}, \frac{1}{2} \right) \right) : Q_1 \in {\cal Q}_n \right\}.
\end{align}
For instance, when 
\begin{align}
(Q_0,Q_1,Q_2) = \left( \left( \frac{2}{3}, \frac{1}{3}  \right), \left( \frac{1}{2}, \frac{1}{2} \right), \left( \frac{1}{2}, \frac{1}{2}  \right) \right),
\end{align}
then $e(Q_0,Q_1,Q_2) = \{0,1\}$; when 
\begin{align}
(Q_0,Q_1,Q_2) = \left( \left( \frac{2}{3}, \frac{1}{3}  \right), \left( \frac{2}{3}, \frac{1}{3} \right), \left( \frac{1}{2}, \frac{1}{2}  \right) \right),
\end{align}
then $e(Q_0,Q_1,Q_2) = \{1,2\}$; when
\begin{align}
(Q_0,Q_1,Q_2) = \left( \left( \frac{2}{3}, \frac{1}{3}  \right), \left( \frac{1}{3}, \frac{2}{3} \right), \left( \frac{1}{2}, \frac{1}{2}  \right) \right),
\end{align}
then $e(Q_0,Q_1,Q_2) = \{1\}$.
\end{example}

Lemma \ref{lemma:compatible} guarantees that 
the hyperedge $e(Q_0,\ldots,Q_{|{\cal Y}|-1})$ which is compatible with 
$(Q_0,\ldots,Q_{|{\cal Y}|-1}) \in {\cal Q}_n^{(i)}({\cal S},\bm{x},\bm{y})$
includes $x_i$.
In the proof of Theorem \ref{theorem:non-full-support}, 
a hyperedge including $x_i$ 
plays a similar role as
a subset $[x_i]_{\bcX}$
of a partition $\bcX$ including $x_i$ in the case of smooth sources.
Particularly, 
in the achievability proof,
we compute $f_n^\mathsf{t}(\bm{x},\bm{y})$ from (i) a partial information on $x_i$ such that it is in a hyperedge and (ii) additional 
information on the (conditional) marginal type of $\bm{x}$.
To guarantee that the decoder can compute the function value, we need a technical condition on hyperedges.
Before the precise description of the condition, we give an example which shows a key idea behind the definition.

\begin{example}
\label{example:idea_solvability}
Let us first consider a source such that $\cX=\{0,1\}$, $\cY=\{0,1,2\}$, and the support set 
is $\cS=\{(0,1),(0,2),(1,0),(1,2)\}$; cf.~Table \ref{table:example:idea_solvability-1}, where $*$ indicates $(x,y) \notin {\cal S}$.
Using the terminology that will be introduced in Definition \ref{def:solvability}, this is the case such that there is no ``simple loop".
In this case, the joint type $P_{\bm{x}\bm{y}}$ can be determined from marginal types $P_{\bm{x}}$ and $P_{\bm{y}}$ for any $(\bx,\by)\in\cS^n$. 
Indeed, we have 
$nP_{\bx\by}(0,1) = nP_{\by}(1)$, 
$nP_{\bx\by}(0,2) = n(P_{\bx}(0)-P_{\by}(1))$, 
$nP_{\bx\by}(1,0) = nP_{\by}(0)$, and
$nP_{\bx\by}(1,2) = n(P_{\bx}(1)-P_{\by}(0))$.
Thus, for any function $f$ on $\cX\times\cY$, the value $f_n^\mathsf{t}(\bm{x},\bm{y})$ can be computed from $P_{\bm{x}}$ and $P_{\bm{y}}$.
In this sense, this case is ``solvable".
\begin{table}[tb]
\centering{
\caption{The Case without simple loop in Example \ref{example:idea_solvability}} \label{table:example:idea_solvability-1}
\begin{tabular}{c|ccc}
 $x\setminus y$ & $0$ & $1$ & $2$  \\\hline
 $0$ & * &  &  \\
 $1$ &  & * &  \\
\end{tabular}
}
\end{table}

On the other hand, let us consider the case where the support $\cS$ and the function $f$ are given as 
Table \ref{table:example:idea_solvability-2}.
Using the terminology that will be introduced in Definition \ref{def:solvability}, this is the case 
such that there exists a ``balanced simple loop"; see also Example \ref{example:simple-loops}.
In this case, for any $(\bx,\by)\in\cS^n$, letting $a=nP_{\bx\by}(0,0)$, we have 
\begin{align}
 nP_{\bx\by}(0,2) &=nP_{\bx}(0)-a,\label{eq1:example:idea_solvability}\\
 nP_{\bx\by}(1,0) &=nP_{\by}(0)-a,\\
 nP_{\bx\by}(1,1) &=nP_{\bx}(1)-nP_{\by}(0)+a,\label{eq2:example:idea_solvability}\\
 nP_{\bx\by}(2,1) &=nP_{\by}(1)-nP_{\bx}(1)+nP_{\by}(0)-a,\\
 nP_{\bx\by}(2,2) &=nP_{\by}(2)-nP_{\bx}(0)+a.
\end{align}
Since the value $a$ is unknown, which stems from the fact that there is a simple loop,
$P_{\bx\by}$ cannot be determined from $P_{\bx}$ and $P_{\by}$. Nevertheless, since the simple loop 
is balanced in the sense that $+a$ and $-a$ cancel for each function value, 
we can compute the type $P_{f_n(\bx,\by)}$ of the function values from $P_{\bx}$ and $P_{\by}$. 
Indeed, from \eqref{eq1:example:idea_solvability} and \eqref{eq2:example:idea_solvability}, we have $nP_{f_n(\bx,\by)}(1)=nP_{\bx}(0)+nP_{\bx}(1)-nP_{\by}(0)$, which does not depend on unknown value $a$.
We can also determine $nP_{f_n(\bx,\by)}(2)$ and $nP_{f_n(\bx,\by)}(3)$ similarly.
In this sense, this case is also ``solvable".

\begin{table}[tb]
\centering{
\caption{The case with balanced simple loop in Example \ref{example:idea_solvability}} \label{table:example:idea_solvability-2}
\begin{tabular}{c|ccc}
 $x\setminus y$ & $0$ & $1$ & $2$  \\\hline
 $0$ & 2 & * & 1 \\
 $1$ & 3 & 1 & * \\
 $2$ & * & 2 & 3
\end{tabular}
}
\end{table}
\end{example}

Motivated by the idea given in Example \ref{example:idea_solvability}, we introduce the concept of solvability of hyperedges.
The solvable hyperedge gives a sufficient condition such that 
the type of function values can be determined from 
marginal types, which is guaranteed by
Lemma \ref{lemma:marginal_condition} below.

\begin{definition}[Solvable Hyperedge]
\label{def:solvability}
For a given ${\cal A} \times {\cal B} \subseteq {\cal X} \times {\cal Y}$, a subset of the form:
\begin{align}
\{ (a_0,b_0),(a_0,b_1),(a_1,b_1),(a_1,b_2),\ldots,(a_{m-2},b_{m-1}), (a_{m-1},b_{m-1}), (a_{m-1},b_0) \} \subseteq ({\cal A} \times {\cal B}) \cap {\cal S}
\end{align}
with $a_i \neq a_j$ and $b_i \neq b_j$ for $i \neq j$, is called a {\em simple loop}.
For a given simple loop and each $v \in {\cal V}$, let 
\begin{align}
{\cal I}_+(v) &:= \{ 0 \le i \le m-1 : f(a_i,b_i) = v \}, \\
{\cal I}_-(v) &:= \{ 0 \le i \le m-1 : f(a_i, b_{i+1~\mathrm{mod}~m}) = v\}
\end{align}
be the set of incremental positions and the set decremental positions in the simple loop, respectively.  
Then, we say that ${\cal A} \times {\cal B}$ is {\em solvable} for $({\cal S}, f)$ if, for any simple loop of ${\cal A}\times {\cal B}$, 
the {\em balanced} condition
\begin{align}
|{\cal I}_+(v)| = |{\cal I}_-(v)|,~~~~\forall v \in {\cal V}
 \label{eq:balanced-condition}
\end{align}
holds.\footnote{When either $|{\cal A}| \le 1$ or $|{\cal B}| \le1$, then ${\cal A} \times {\cal B}$ is trivially solvable since there is no simple loop.} 
We say that $e \subseteq {\cal X}$ is solvable hyperedge if $e \times {\cal Y}$ is solvable for $({\cal S}, f)$.
The set of all maximal solvable hyperedges for $({\cal S}, f)$ is denoted by ${\cal E}({\cal S},f)$.
\end{definition}

\begin{remark}
When $f$ is the identity function, we can verify that ${\cal A} \times {\cal B}$ is solvable for $({\cal S},f)$ if and only if
it does not contain any simple loop.
\end{remark}

\begin{example} \label{example:simple-loops}
Let us consider function $f:{\cal X}\times {\cal Y} \to {\cal V}=\{0,\ldots,4\}$ shown in 
Table \ref{table:example-simple-loop}.
There are two simple loops for this function:
\begin{align}
& \{ (0,1), (0,3), (2,3), (2,2), (1,2), (1,1) \}, \\ 
& \{ (0,0), (0,4), (3,4), (3,0) \},
\end{align}
which are described in Tables \ref{table:example-simple-loop2} and \ref{table:example-simple-loop3}
with subscripts $\pm$, where $+$ and $-$ indicate incremental and decremental positions, respectively. 
As we can find from the tables, the balanced condition \eqref{eq:balanced-condition} is satisfied for both the simple loops. 
Thus, $\{ 0,1,2,3\} \times \{0,1,2,3,4\}$ is solvable in this case.
\begin{table}[tb]
\begin{minipage}{.3\textwidth}
\centering{
\caption{Function of Example \ref{example:simple-loops}}\label{table:example-simple-loop}
\begin{tabular}{c|ccccc}
 $x\setminus y$& $0$ & $1$ & $2$ & $3$ & $4$ \\\hline
 $0$ & $4$ & $2$ & * & $1$ & $0$ \\
 $1$ & * & $3$ & $1$ & * & * \\
 $2$ & * & * & $2$ & $3$ & * \\
 $3$ & $4$ & * & * & * & $0$ \\
\end{tabular}
}
\end{minipage}
~
\begin{minipage}{.3\textwidth}
\centering{
\caption{Simple Loop 1}\label{table:example-simple-loop2}
\begin{tabular}{c|ccccc}
 $x\setminus y$& $0$ & $1$ & $2$ & $3$ & $4$ \\\hline
 $0$ & $4$ & $2_+$ & * & $1_-$ & $0$ \\
 $1$ & * & $3_-$ & $1_+$ & * & * \\
 $2$ & * & * & $2_-$ & $3_+$ & * \\
 $3$ & $4$ & * & * & * & $0$ \\
\end{tabular}
}
\end{minipage}
~
\begin{minipage}{.3\textwidth}
\centering{
\caption{Simple Loop 2}\label{table:example-simple-loop3}
\begin{tabular}{c|ccccc}
 $x\setminus y$& $0$ & $1$ & $2$ & $3$ & $4$ \\\hline
 $0$ & $4_+$ & $2$ & * & $1$ & $0_-$ \\
 $1$ & * & $3$ & $1$ & * & * \\
 $2$ & * & * & $2$ & $3$ & * \\
 $3$ & $4_-$ & * & * & * & $0_+$ \\
\end{tabular}
}
\end{minipage}
\end{table}
\end{example}

\begin{example} \label{example-solvable-card-game}
Let us consider the function given by \eqref{eq:function-card-problem}; see also
Table \ref{table:example-card-game}. In this case, $\{0,1\}$, $\{0,2\}$, $\{1,2\}$ are solvable 
hyperedges since there is no simple loop. However, $\{0,1,2\}$ is not solvable hyperedge since
the simple loop described in Table \ref{table:example-card-game} violates \eqref{eq:balanced-condition}.
\begin{table}[tb]
\centering{
\caption{Function table of \eqref{eq:function-card-problem}} \label{table:example-card-game}
\begin{tabular}{c|ccc}
 $x\setminus y$ & $0$ & $1$ & $2$  \\\hline
 $0$ & * & $1_+$ & $1_-$ \\
 $1$ & $0_-$ & * & $1_+$ \\
 $2$ & $0_+$ & $0_-$ & *  \\
\end{tabular}
}
\end{table}
\end{example}

\begin{lemma}
\label{lemma:marginal_condition}
Suppose that ${\cal A} \times {\cal B} \subseteq {\cal X} \times {\cal Y}$ is solvable for $({\cal S},f)$. 
Then, for any integer $n$, $\bm{x} \in {\cal A}^n$ and $\bm{y} \in {\cal B}^n$ satisfying $(x_i,y_i) \in {\cal S}$  for all $1 \le i \le n$,
the type $f_n^{\mathsf{t}}(\bm{x},\bm{y})$ of $f_n(\bm{x},\bm{y}) = (f(x_1,y_1),\ldots,f(x_n,y_n))$ can be uniquely determined from
the marginal types $P_{\bm{x}}$ and $P_{\bm{y}}$; more precisely, any joint types $P_{\bar{X}\bar{Y}}^{(1)}, P_{\bar{X}\bar{Y}}^{(2)} \in {\cal P}_n({\cal A} \times {\cal B})$
with $\mathsf{supp}(P_{\bar{X}\bar{Y}}^{(i)}) \subseteq {\cal S}$, $i=1,2$ satisfying 
\begin{align}
\sum_{y \in {\cal B}} P_{\bar{X}\bar{Y}}^{(i)}(x,y) &= P_{\bm{x}}(x),~~~\forall x \in {\cal A},\label{eq:marginal-condition-x} \\
\sum_{x \in {\cal A}} P_{\bar{X}\bar{Y}}^{(i)}(x,y) &= P_{\bm{y}}(y),~~~\forall y \in {\cal B} \label{eq:marginal-condition-y}
\end{align}
must satisfy 
\begin{align}
\sum_{(x,y) \in {\cal A}\times {\cal B}: \atop f(x,y)=v} P_{\bar{X}\bar{Y}}^{(1)}(x,y) = \sum_{(x,y) \in {\cal A}\times {\cal B}: \atop f(x,y) = v} P_{\bar{X}\bar{Y}}^{(2)}(x,y),~~~\forall v \in {\cal V}.
 \label{eq:unique-solution-condition}
\end{align}
\end{lemma}

The following lemma gives a connection between compatible hyperedges and solvable hyperedges, which 
will be used in the converse proof of Theorem \ref{theorem:non-full-support}. 

\begin{lemma}
\label{lemma:connection-compatible-solvable}
For a given $(Q_0,\ldots, Q_{|{\cal Y}|-1}) \in {\cal Q}_n^{|{\cal Y}|}$, 
the compatible hyperedge $e(Q_0,\ldots,Q_{|{\cal Y}|-1})$ is solvable. Furthermore, there exists $\tilde{e} \in {\cal E}({\cal S},f)$
satisfying $e(Q_0,\ldots,Q_{|{\cal Y}|-1}) \subseteq \tilde{e}$.
\end{lemma}

\subsection{Coding Theorem}

\begin{theorem} \label{theorem:non-full-support}
For given $P_{XY}$ 
with ${\cal S} = \mathsf{supp}(P_{XY})$ and $f:{\cal X} \times{\cal Y} \to {\cal V}$,
let ${\cal G} = ({\cal X},{\cal E})$ be the hypergraph such that ${\cal E} = {\cal E}({\cal S},f)$ is the set of all maximal solvable 
hyperedges for $({\cal S},f)$. Then, the optimal rate for computing the type of symbol-wise function $f$ for
i.i.d. source $(\bm{X},\bm{Y})$ distributed according to $P_{XY}$ is given by
\begin{align}
R(\bm{X}|\bm{Y}|\bm{f}^{\mathsf{t}}) = H_{{\cal G}}(X|Y).
\end{align}
\end{theorem}

To illustrate the utility of Theorem \ref{theorem:non-full-support}, let us consider the following 
example from \cite{OrlitskyRoche01} (see also \cite{ElGamalKim}). 
\begin{example} \label{example:comparison-main-theorem}
Consider an $n$-round online game, where in each round Alice and Bob each select one card without replacement
from a virtual hat with three cards labeled $0, 1,2$. The one with larger number wins. Let $X^n$ be Alice's outcome
and $Y^n$ be Bob's outcome. This situation is described by ${\cal Y} = {\cal X}$, 
${\cal S} = \{ (x,y) : x,y \in \{0,1,2\},  x \neq y \}$, $P_{XY}(x,y) = \frac{1}{6}$ for $(x,y) \in {\cal S}$, and the function $f$ defined by \eqref{eq:function-card-problem}.
If Bob would like to know who won in each round, it suffices for Alice to send a message at rate $\frac{2}{3} h\left(\frac{1}{4}\right)$,
which is optimal \cite{OrlitskyRoche01}. 

Now, suppose that Bob does not care who won in each round; instead, he is only interested in
the total number of rounds he won. Then, Theorem \ref{theorem:non-full-support} (see also Example \ref{example:conditional-hypergraph-entropy} and Example \ref{example-solvable-card-game}) says that 
it suffices for Alice to send a message at rate $\frac{1}{2}$, which is optimal. 
\end{example}

\begin{remark}
For ${\cal X} = {\cal Y} = \{0,1,2\}$, ${\cal S} = \{ (x,y) : x,y \in \{0,1,2\}, x \neq y\}$, and the identity function $f^\mathsf{id}$,
the set of all solvable hyperedges is given by $\{ \{0,1\}, \{0,2\}, \{1,2\}\}$, which is the same as Example \ref{example-solvable-card-game}.
Thus, the optimal rate for computing the type of \eqref{eq:function-card-problem} is the same as
the optimal rate for computing the joint type.
This is in contrast to the fact that the optimal rate for computing \eqref{eq:function-card-problem} symbol-wisely is strictly smaller than
the optimal rate for computing the identity function symbol-wisely, i.e., the Slepian-Wolf rate.
\end{remark}

\subsection{Proof of Theorem}

We first prove the converse part
\begin{align}
 R(\bX|\bY|\bmf^{\mathsf{t}})\geq H_{{\cal G}}(X|Y)
\label{eq_converse:proof:theorem:non-full-support}
\end{align}
and then prove the direct part
\begin{align}
 R(\bX|\bY|\bmf^{\mathsf{t}})\leq H_{{\cal G}}(X|Y).
\label{eq_direct:proof:theorem:non-full-support}
\end{align}

\begin{IEEEproof}
[Converse part]
Fix $\varepsilon>0$ arbitrarily, and let $(\varphi_n,\psi_n)$ be a code with size $\lvert\cM_n\rvert$ satisfying
\begin{align}
 \Pr\left(\psi_n(\varphi_n(X^n),Y^n)\neq f_n^{\mathsf{t}}(X^n,Y^n)\right)&=\sum_{\bx,\by}P_{X^nY^n}(\bx,\by)\bm{1}\left[\psi_n(\varphi_n(\bx),\by)\neq f_n^{\mathsf{t}}(\bx,\by)\right]\\
&\leq\varepsilon.
\end{align}

First, by a similar manner to the proof of the converse part of Theorem \ref{thm:full_support}, we prove that, for any $i\in[1:n]$,
\begin{align}
  \Pr\left(
\psi_n(\varphi_n(X^n),bY^{(-i)})= f_n^{\mathsf{t}}(X^n,bY^{(-i)}), \forall b\in\cY\text{ s.t. }(X_i,b)\in\cS
\right)\geq 1-\frac{\lvert\cY\rvert}{q_*}\varepsilon
\label{eq2_converse:proof:theorem:non-full-support}
\end{align}
where 
\begin{align}
 q_* := \min_{(x,y)\in\cS}P_{XY}(x,y).
\end{align}
Indeed, for every 
$\hat{b}$ ($0\leq\hat{b}\leq\lvert\cY\rvert-1$)
and $i\in[1:n]$, we have
\begin{align}
\lefteqn{\Pr\left([\psi_n(\varphi_n(X^n),\pi_i^{\hat{b}}(Y^n))\neq f_n^{\mathsf{t}}(X^n,\pi_i^{\hat{b}}(Y^n))]
\wedge
[(X_i,\pi^{\hat{b}}(Y_i))\in\cS]
\right)}\nonumber\\
&=\sum_{\bx,\by}P_{X^nY^n}(\bx,\by)\bm{1}\left[
[\psi_n(\varphi_n(\bx),\pi_i^{\hat{b}}(\by))\neq f_n^{\mathsf{t}}(\bx,\pi_i^{\hat{b}}(\by))]
\wedge
[(x_i,\pi^{\hat{b}}(y_i))\in\cS]
\right]\\
&\leq \sum_{\bx,\by}\frac{1}{q_*}P_{X^nY^n}(\bx,\pi_i^{\hat{b}}(\by))\bm{1}\left[
[\psi_n(\varphi_n(\bx),\pi_i^{\hat{b}}(\by))\neq f_n^{\mathsf{t}}(\bx,\pi_i^{\hat{b}}(\by))]
\wedge
[(x_i,\pi^{\hat{b}}(y_i))\in\cS]
\right]\\
&\leq \frac{\varepsilon}{q_*}\label{eq:proof:thm:non-full-support}
\end{align}
where $\pi\colon\cY\to\cY$ is the permutation such that $y\mapsto y+1 (\bmod \lvert\cY\rvert)$.
Thus, for any $i\in[1:n]$, we have
\begin{align}
\lefteqn{
 \Pr\left(
\psi_n(\varphi_n(X^n),bY^{(-i)})\neq f_n^{\mathsf{t}}(X^n,bY^{(-i)}) \text{ for some $b\in\cY$ s.t. }(X_i,b)\in\cS
\right)}\nonumber\\
&= \Pr\left(\exists b\in\cY\text{ s.t. }\psi_n(\varphi_n(X^n),bY^{(-i)})\neq f_n^{\mathsf{t}}(X^n,bY^{(-i)})\text{ and }(X_i,b)\in\cS\right)\\
&= \sum_{\bx,\by}P_{X^nY^n}(\bx,\by)\bm{1}\left[\exists b\in\cY\text{ s.t. }\psi_n(\varphi_n(\bx),b\by^{(-i)})\neq f_n^{\mathsf{t}}(\bx,b\by^{(-i)})\text{ and }(x_i,b)\in\cS\right]\\
&= \sum_{\bx,\by}P_{X^nY^n}(\bx,\by)\bm{1}\left[\exists 0\leq \hat{b}\leq\lvert\cY\rvert-1\text{ s.t. }\psi_n(\varphi_n(\bx),\pi_i^{\hat{b}}(\by))\neq f_n^{\mathsf{t}}(\bx,\pi_i^{\hat{b}}(\by))\text{ and }(x_i,\pi^{\hat{b}}(y_i))\in\cS\right]\\
&= \Pr\left(\exists 0\leq \hat{b}\leq\lvert\cY\rvert-1\text{ s.t. }[\psi_n(\varphi_n(X^n),\pi_i^{\hat{b}}(Y^n))\neq f_n^{\mathsf{t}}(X^n,\pi_i^{\hat{b}}(Y^n))]
\wedge
[(X_i,\pi^{\hat{b}}(Y_i))\in\cS]
\right)\\
&\leq \frac{\lvert\cY\rvert}{q_*}\varepsilon.
\end{align}
This implies \eqref{eq2_converse:proof:theorem:non-full-support}.

On the other hand, for each $i\in[1:n]$, let us define a random variable $W_i$ as follows.
For each $b\in\cY$, let
\begin{align}
 Q_b = \psi_n(\varphi_n(X^n),bY^{(-i)}).
\end{align}
Then, we set $W_i=w$ for a hyperedge $w\in\cE$ satisfying $w\supseteq e(Q_0,\dots,Q_{\lvert\cY\rvert-1})$,
where existence of such a hyperedge is guaranteed by Lemma \ref{lemma:connection-compatible-solvable}.\footnote{If there are more than one such $w\in\cE$ then we pick arbitrary one.}

Note that $W_i$ satisfies the Markov chain $W_i\markov X_i\markov Y_i$, since $W_i$ is determined from $\varphi_n(X^n)$, $Y_1^{i-1}$, and $Y_{i+1}^n$.
Furthermore, from Lemma \ref{lemma:compatible} and \eqref{eq2_converse:proof:theorem:non-full-support}, it is not hard to see that
\begin{align}
 \Pr(X_i\in W_i)\geq 1-\gamma
\end{align}
where
\begin{align}
 \gamma:=\frac{\lvert\cY\rvert}{q_*}\varepsilon.
\end{align}
Now, let us introduce a new quantity
\begin{align}
H_{\cG}^{\gamma}(X|Y) &:= \min_{\substack{
W \markov X \markov Y\\
X \in W \text{ with prob.}\geq 1-\gamma}} I(W \wedge X|Y) \\
&= \min\bigg\{ I(W \wedge X|Y) : W \text{ takes values in } \cE, W \markov X \markov Y, \text{ and } \Pr( X \in W) \geq 1-\gamma \bigg\}.
\end{align}
Then, the random variable $W_i$ defined above satisfies that
\begin{align}
I(X_i\wedge W_i|Y_i)\geq H_{\cG}^{\gamma}(X|Y)\quad \forall i\in[1:n].
\end{align}
Hence, by the standard argument, we have the following chain of inequalities
\begin{align}
 \log\lvert\cM_n\rvert &\geq H(\varphi_n(X^n))\\
&\geq H(\varphi_n(X^n)|Y^n)\\
&= I(X^n\wedge\varphi_n(X^n)|Y^n)\\
&= \sum_{i=1}^n I(X_i\wedge \varphi_n(X^n)|Y^n,X^{i-1})\\
&= \sum_{i=1}^n \left[ H(X_i|Y^n,X^{i-1})-H(X_i|\varphi_n(X^n),Y^n,X^{i-1}) \right] \\
&= \sum_{i=1}^n \left[ H(X_i|Y_i)-H(X_i|\varphi_n(X^n),Y^n,X^{i-1}) \right] \\
&\geq \sum_{i=1}^n \left[ H(X_i|Y_i)-H(X_i|\varphi_n(X^n),Y^n) \right] \\
&\stackrel{\text{(a)}}{\geq} \sum_{i=1}^n \left[ H(X_i|Y_i)-H(X_i|W_i,Y_i) \right] \\
&= \sum_{i=1}^n I(X_i\wedge W_i|Y_i)\\
&\geq nH_{\cG}^{\gamma}(X|Y)
\label{eq3_converse:proof:theorem:non-full-support}
\end{align}
where the inequality (a) follows from the fact that $W_i$ is determined from $\varphi_n(X^n)$, $Y_1^{i-1}$, and $Y_{i+1}^n$.

Eq.~\eqref{eq3_converse:proof:theorem:non-full-support} implies that
\begin{align}
  R(\bX|\bY|\bmf^{\mathsf{t}})\geq H_{\cG}^{\gamma}(X|Y).
\end{align}
Since we can choose $\varepsilon$ arbitrarily small and $\gamma\to 0$ as $\varepsilon\to 0$, we have
\begin{align}
  R(\bX|\bY|\bmf^{\mathsf{t}})&\geq \sup_{\gamma>0}H_{\cG}^{\gamma}(X|Y)\\
&=\lim_{\gamma\downarrow 0}H_{\cG}^{\gamma}(X|Y)\\
&=H_{\cG}(X|Y)
\end{align}
where the last equality holds from 
the compactness of the set of conditional probabilities $P_{W|X}$
and
the continuity of the conditional mutual information $I(W\wedge X|Y)$.
Hence, we have \eqref{eq_converse:proof:theorem:non-full-support}.
\end{IEEEproof}

\begin{IEEEproof}
[Direct part]
The proof of the direct part is divided into two parts: in the first part, the encoder sends 
a quantized version of $X^n$ to the decoder; in the second part, the encoder additionally sends the (conditional) marginal type,
and the decoder computes the function value by using solvability of hyperedges.
Since the first part is the standard argument of the Wyner-Ziv coding, 
we only provide a sketch (see \cite[Sec. 11.3.1]{ElGamalKim} for the detail).

Let $P_{W|X}$ be a test channel that attain $H_{\cal G}(X|Y)$, and let $P_{WXY}$ be the joint
distribution induced by the test channel and $P_{XY}$. Fix arbitrary $\varepsilon > \varepsilon^\prime >0$, and let
\begin{align}
{\cal T}_{\varepsilon}^n(WY) := \left\{ (\bm{w},\bm{y}) : | P_{\bm{w}\bm{y}}(w,y) - P_{WY}(w,y) | \le \varepsilon P_{WY}(w,y),~\forall (w,y) \in {\cal E}\times {\cal Y}  \right\}
\end{align}
be the set of all $\varepsilon$-typical sequences; ${\cal T}_{\varepsilon^\prime}^n(WX)$ is defined similarly. 
We use the so-called quantize-bin scheme. For codebook generation, we randomly and independently generate $2^{n\tilde{R}}$ codewords
$\bm{w}(\ell)$, $\ell \in [1:2^{n \tilde{R}}]$, each according to $P_W^n$. Then, we partition the set of indices $\ell \in [1:2^{n\tilde{R}}]$ into
equal-size bins ${\cal B}(m)$, $m\in [1:2^{nR}]$. For encoding, given $\bm{x}$, the encoder finds, if exists, an index $\ell(\bm{x})$ such that
$(\bm{w}(\ell(\bm{x})), \bm{x}) \in {\cal T}_{\varepsilon^\prime}^n(WX)$, and sends the bin index $m(\bm{x})$ satisfying $\ell(\bm{x}) \in {\cal B}(m(\bm{x}))$. 
For decoding, upon receiving message $m$, the decoder finds, if exists, the unique index $\hat{\ell}(m,\bm{y}) \in {\cal B}(m)$ such that
$(\bm{w}(\hat{\ell}(m,\bm{y})), \bm{y}) \in {\cal T}_{\varepsilon}^n(WY)$. Then, if $\tilde{R} > I(W \wedge X) +\delta(\varepsilon^\prime)$ and 
$\tilde{R} - R < I(W \wedge Y) - \delta(\varepsilon)$, where $\delta(\varepsilon), \delta(\varepsilon^\prime) \to 0$ as $\varepsilon,\varepsilon^\prime \to 0$, the
following performance is guaranteed:
\begin{align}
\lim_{n\to\infty} \Pr\left( (\bm{w}(\ell(X^n)), X^n) \notin {\cal T}_{\varepsilon^\prime}^n(WX) \mbox{ or } \hat{\ell}(m(X^n), Y^n) \neq \ell(X^n)  \right) = 0.
 \label{eq:error-quantize-bin}
\end{align}

In addition to message $m(\bm{x})$, the encoder also sends the marginal type $Q_{\bar{W}} = P_{\bm{w}(\ell(\bm{x}))} \in {\cal P}_n({\cal E})$ and
the conditional type $Q_{\bar{X}|\bar{W}} = P_{\bm{x}|\bm{w}(\ell(\bm{x}))} \in {\cal P}_n({\cal X}|{\cal E})$.
Then, upon receiving $Q_{\bar{W}}$ and $Q_{\bar{X}|\bar{W}}$, the decoder declares an error if $Q_{\bar{W}} \neq P_{\bm{w}(\hat{\ell}(m,\bm{y}))}$.
Otherwise, the decoder finds, if exists, the unique conditional type $\hat{Q}_{\bar{V}|\bar{W}} \in {\cal P}_n({\cal V}|{\cal E})$
that is compatible with $Q_{\bar{X}|\bar{W}}$ and $Q_{\bar{Y}|\bar{W}} = P_{\bm{y}|\bm{w}(\hat{\ell}(m,\bm{y}))}$ in the following sense:
for every conditional joint type $Q_{\bar{X}\bar{Y}|\bar{W}}$ with $\mathsf{supp}(Q_{\bar{X}\bar{Y}|\bar{W}}(\cdot,\cdot|w)) \subseteq {\cal S}$, $\forall w \in {\cal E}$
such that the marginals are $Q_{\bar{X}|\bar{W}}$ and $Q_{\bar{Y}|\bar{W}}$ respectively, it holds that
\begin{align}
\hat{Q}_{\bar{V}|\bar{W}}(v|w) = \sum_{(x,y) : \atop f(x,y) = v} Q_{\bar{X}\bar{Y}|\bar{W}}(x,y|w)
 \label{eq:compatible-for-given-w}
\end{align}
for every $v \in {\cal V}$ and $w \in {\cal E}$.
Then, the decoder outputs 
\begin{align}
\hat{Q}_{\bar{V}}(\cdot ) = \sum_{w \in {\cal E}} Q_{\bar{W}}(w) \hat{Q}_{\bar{V}|\bar{W}}(\cdot|w) \in {\cal P}_n({\cal V})
\end{align}
as an estimate of the function value $f_n^{\mathsf{t}}(\bm{x},\bm{y})$.

Fix arbitrary $(\bm{x},\bm{y})$ satisfying $(x_i,y_i) \in {\cal S}$ for all $1 \le i \le n$. We claim that 
the estimate $\hat{Q}_{\bar{V}}$ coincide with the function value $f_n^{\mathsf{t}}(\bm{x},\bm{y})$ whenever 
$(\bm{w}(\ell(\bm{x})), \bm{x}) \in {\cal T}_{\varepsilon^\prime}^n(WX)$ and $\hat{\ell}(m(\bm{x}), \bm{y}) = \ell(\bm{x})$. 
In fact, note that $(\bm{w}(\ell(\bm{x})), \bm{x}) \in {\cal T}_{\varepsilon^\prime}^n(WX)$ implies 
$x_i \in \bm{w}(\ell(\bm{x}))_i$ for every $i \in [1:n]$, where $\bm{w}(\ell(\bm{x}))_i$ is the $i$th component of $\bm{w}(\ell(\bm{x}))$.
Thus, by applying Lemma \ref{lemma:marginal_condition} for each $w \in {\cal E}$ with
${\cal A} = w$, ${\cal B} = {\cal Y}$, and $n = |\{ i \in [1:n] : \bm{w}(\ell(\bm{x}))_i = w\}|$, existence and uniqueness of 
$\hat{Q}_{\bar{V}|\bar{W}}$ satisfying \eqref{eq:compatible-for-given-w} is guaranteed, and it satisfies $\hat{Q}_{\bar{V}|\bar{W}} = P_{f_n(\bm{x},\bm{y}) | \bm{w}(\ell(\bm{x}))}$,
which implies $\hat{Q}_{\bar{V}} = f_n^\mathsf{t}(\bm{x},\bm{y})$. 
This claim together with \eqref{eq:error-quantize-bin} imply that the error probability of computing $f_n^\mathsf{t}(X^n,Y^n)$ at the decoder vanishes asymptotically.
Since types $Q_{\bar{W}}$ and $Q_{\bar{X}|\bar{W}}$ can be sent with asymptotically zero-rate,
the total rate is bounded by $H_{\cal G}(X|Y) + \delta(\varepsilon) + \delta(\varepsilon^\prime)$. 
Since $\delta(\varepsilon), \delta(\varepsilon^\prime)$ can be made arbitrarily small, 
we have \eqref{eq_direct:proof:theorem:non-full-support}.
\end{IEEEproof}

\section{Conclusion}\label{sec:counclusion}

In this paper, we developed a new method to show converse bounds on
the distributed computing problem. By using the proposed method, we characterized 
the optimal rate of distributed computing for some classes of functions that are difficult to
be handled by previously known methods. 
The key idea of our method is, from the nature of distributed computing and the structure of the function to be computed, 
to identify information that is inevitably conveyed to the decoder. We believe that our method is useful for characterizing
the optimal rate for more general classes of functions, which will be studied in a future work. 

Another important problem to be studied is
the case where both the sources $X^n$ and $Y^n$ are encoded by two separate encoders. 
In such a problem, a difficulty is to derive a bound on the sum rate. 
Extending the method developed in this paper to such a problem is an important future research agenda. 

\appendices

\section{Technical Lemmas}\label{appendix:misc}

The following lemma says that if there exists a code with
small symbol error probability, then, by sending additional message of negligible rate, we can boost that code
so that block error probability is small. For given $0 < \beta < 1/2$, let
\begin{align} 
\nu_n(\beta) := \sum_{i=0}^{\lceil n \beta \rceil -1} (|{\cal X}|-1)^i {n \choose i} 
\le n |{\cal X}|^{n \beta} 2^{n h(\beta)}
\end{align}
be the size of Hamming ball of radius $\lceil n \beta \rceil -1$ on ${\cal X}^n$.

\begin{lemma} \label{lemma:boosting}
Suppose that $(X^n,W^n)$ on ${\cal X}^n \times {\cal X}^n$ satisfies
\begin{align}
\Pr\left( \frac{1}{n} d_H(X^n, W^n) \ge \beta \right) \le \varepsilon_n. 
\end{align}
Then, there exists an encoder $\kappa_n:{\cal X}^n \to {\cal K}_n$ with $|{\cal K}_n| \le 2^{n \delta}$ and 
a decoder $\tau_n:{\cal K}_n \times {\cal X}^n \to {\cal X}^n$ such that 
\begin{align}
\Pr\left( \tau_n(\kappa_n(X^n),W^n) \neq X^n \right) \le \varepsilon_n + \nu_n(\beta) 2^{-n\delta}.
\end{align}
\end{lemma}
\begin{proof}
For an encoder $\kappa_n$, we use the random binning. Given $k_n \in {\cal K}_n$ and $\bm{w} \in {\cal X}^n$, the decoder finds (if exists) a unique $\hat{\bm{x}}$ such that 
\begin{align}
\hat{\bm{x}} \in {\cal T}_\beta^n(\bm{w}) := \left\{ \bm{x} : d_H(\bm{x},\bm{w}) < n \beta \right\}  
\end{align}
and $\kappa_n(\hat{\bm{x}}) = k_n$.
Note that 
\begin{align}
|{\cal T}_\beta^n(\bm{w})| \le \nu_n(\beta),~\forall \bm{w} \in {\cal X}^n.
\end{align}
Then, by the standard argument (cf.~\cite[Lemma 7.2.1]{Han-spectrum}), the error probability averaged over the random binning is bounded as
\begin{align}
\lefteqn{ \mathbb{E}_{\kappa_n}\left[ \Pr\left( \tau_n(\kappa_n(X^n),W^n) \neq X^n \right) \right] } \\
&\le \Pr\left( \frac{1}{n} d_H(X^n, W^n) \ge \beta \right) \\ 
&~~~+ \sum_{\bm{x},\bm{w}} P_{X^n W^n}(\bm{x},\bm{w}) \sum_{\hat{\bm{x}} \in {\cal T}_\beta^n(\bm{w}) \atop \hat{\bm{x}} \neq \bm{x}}
 \Pr\left( \kappa_n(\hat{\bm{x}}) \neq \kappa_n(\bm{x}) \right) \\
&\le \varepsilon_n + \sum_{\bm{x},\bm{w}} P_{X^n W^n}(\bm{x},\bm{w}) \sum_{\hat{\bm{x}} \in {\cal T}_\beta^n(\bm{w}) \atop \hat{\bm{x}} \neq \bm{x}} \frac{1}{|{\cal K}_n|} \\
&\le \varepsilon_n + \nu_n(\beta) 2^{-n\delta}.
\end{align}
\end{proof}

The following lemma is also used in the main text; it is a slight modification of
the standard converse of the Slepian-Wolf coding (cf.~\cite[Lemma 7.2.2]{Han-spectrum}), where $X^n$
is replaced by a function value $g_n(X^n)$.

\begin{lemma} \label{lemma:converse-function-X}
For a given $(X^n,Y^n)$ and a function $g_n:{\cal X}^n \to {\cal Z}_n$, if a code $(\varphi_n,\psi_n)$ 
with size $|{\cal M}_n|$ satisfies
\begin{align}
\Pr\left( \psi_n(\varphi_n(X^n),Y^n) \neq g_n(X^n) \right) \le \varepsilon_n,
\end{align}
then it holds that 
\begin{align}
\Pr\left( \frac{1}{n} \log \frac{1}{P_{Z_n|Y^n}(Z_n|Y^n)} > \frac{1}{n} \log |{\cal M}_n| + \delta \right) \le \varepsilon_n + 2^{-n\delta},
\end{align}
where $Z_n = g_n(X^n)$.
\end{lemma}
\begin{proof}
By the standard argument, we have
\begin{align}
\lefteqn{ \Pr\left( \frac{1}{n} \log \frac{1}{P_{Z_n|Y^n}(Z_n|Y^n)} > \frac{1}{n} \log |{\cal M}_n| + \delta \right) } \\
&\le \Pr\left( \psi_n(\varphi_n(X^n),Y^n) \neq g_n(X^n) \right) \\
&~~~ + \Pr\bigg( \frac{1}{n} \log \frac{1}{P_{Z_n|Y^n}(Z_n|Y^n)} > \frac{1}{n} \log |{\cal M}_n| + \delta, \\
&~~~ ~\psi_n(\varphi_n(X^n),Y^n) = g_n(X^n) \bigg) \\
&\le \varepsilon_n + \sum_{z_n,m_n,\bm{y}} \sum_{\bm{x} \in g_n^{-1}(z_n) \cap \varphi_n^{-1}(m_n)} P_{Y^n}(\bm{y}) P_{X^n|Y^n}(\bm{x}|\bm{y}) \\
&~~~\times  \mathbf{1}\left[ P_{Z_n|Y^n}(z_n|\bm{y}) < \frac{2^{-n\delta}}{|{\cal M}_n|},~\psi_n(m_n,\bm{y}) = z_n \right] \\
&\le \varepsilon_n + \sum_{z_n,m_n,\bm{y}} P_{Y^n}(\bm{y})   \frac{2^{-n\delta}}{|{\cal M}_n|} 
 \mathbf{1}\left[ \psi_n(m_n,\bm{y}) = z_n \right] \\
&\le \varepsilon_n + 2^{-n\delta},
\end{align}
where the third inequality follows from
\begin{align}
\sum_{\bm{x} \in g_n^{-1}(z_n) \cap \varphi_n^{-1}(m_n)} P_{X^n|Y^n}(\bm{x}|\bm{y})
&\le \sum_{\bm{x} \in g_n^{-1}(z_n)} P_{X^n|Y^n}(\bm{x}|\bm{y}) \\
&= P_{Z_n|Y^n}(z_n|\bm{y}).
\end{align}
\end{proof}

\section{Proof of Proposition \ref{proposition:finest-partition}} \label{appendix:proof-proposition:finest-partition}

\begin{IEEEproof}
First, for given partitions $\bcX_1$ and $\bcX_2$ satisfying Condition (\ref{c1:def:informative}) of Definition \ref{def:informative}, 
it is not difficult to see that their intersection also satisfies Condition (\ref{c1:def:informative}). Thus, there exists the finest 
partition satisfying Condition (\ref{c1:def:informative}), and it suffices to prove that $\bcX$ is the finest one. 

Suppose that there exists a partition $\bcX^\prime$ that is finer than $\bcX$. Then, there exist $a, \hat{a} \in {\cal X}$ such that
$[a]_{\bcX} = [\hat{a}]_{\bcX}$ and $[a]_{\bcX^\prime} \neq [\hat{a}]_{\bcX^\prime}$. Let $\bm{x} \in {\cal X}^n$ be a sequence such that 
$x_i = a$ and $x_j = \hat{a}$ for $i \neq j$. Then, for the permutation $\sigma$ that interchange only $i$ and $j$, we have
\begin{align}
f_n(\sigma(\bm{x}), \bm{y}) = f_n(\bm{x},\bm{y})
\end{align}
holds for every $\bm{y} \in {\cal Y}^n$ since $[\sigma(\bm{x})]_{\bcX} = [\bm{x}]_{\bcX}$ and the partition $\bcX$ satisfies Condition (\ref{c2:def:informative}).
On the other hand, since partition $\bcX^\prime$ satisfies Condition (\ref{c1:def:informative}), 
there exists mapping $\xi_n^{\prime(i)}$ that satisfies \eqref{eq:informative-1} for $\bcX^\prime$. Thus, 
for arbitrarily fixed $\bm{y}$, we have
\begin{align}
[a]_{\bcX^\prime} &= \xi_n^{\prime(i)}\left(\left(f_n(\bm{x},b\by^{(-i)}):b\in\cY\right)\right) \\
&= \xi_n^{\prime(i)}\left(\left(f_n(\sigma(\bm{x}),b\by^{(-i)}):b\in\cY\right)\right) \\
&= [\hat{a}]_{\bcX^\prime},
\end{align}
which is a contradiction.
\end{IEEEproof}

\section{Proof of Propositions \ref{prop:simbolwise}, \ref{prop:type}, and \ref{prop:modsum}}\label{appendix:proof_propositions}
To simplify the notation, we denote $[\cdot]_{\bcX_f}$ by $[\cdot]_f$; e.g.~$[x]_f:=[x]_{\bcX_f}$.

\begin{IEEEproof}
[Proof of Proposition \ref{prop:simbolwise}]
For any $a\in\cX$, $[a]_f$ is uniquely determined from the list $(f(a,b):b\in\cY)$. Hence it is easy to see that
Condition (\ref{c1:def:informative}) of Definition \ref{def:informative} holds.

On the other hand, if $[x_{\sigma(i)}]_f=[x_i]_f$ then
$f(x_{\sigma(i)},y)=f(x,y)$ for all $y\in\cY$. Since $[\sigma(\bx)]_f=[\bx]_f$ means $[x_{\sigma(i)}]_f=[x_i]_f$ for all
$i\in[0:1]$, Condition
(\ref{c2:def:informative}) of Definition \ref{def:informative} holds.
\end{IEEEproof}

\medskip
\begin{IEEEproof}
[Proof of Proposition \ref{prop:type}]
We first show Condition (\ref{c1:def:informative}) of Definition \ref{def:informative} is satisfied.
Fix $a\in\cX$, $(\bx,\by)\in\cX^n\times\cY^n$, and $i\in[1:n]$ arbitrarily, and let
\begin{align}
 Q_b:= f_n^{\mathsf{t}}(a\bx^{(-i)},b\by^{(-i)}),\quad b\in\cY.
\end{align}
Then, we claim that
\begin{align}
 \left(\hat{f}^{\mathsf{t}}(a,b):b\in\cY\right)
\label{eq1:proof_prop:type}
\end{align}
can be uniquely determined from $(Q_b:b\in\cY)$.
In fact, if $(Q_b:b\in\cY)$ is a constant vector, then \eqref{eq1:proof_prop:type} must be $(m,m,\dots,m)$.
Otherwise, find $v_0\in\cV$ such that
\begin{align}
 nQ_b(v_0)<nQ_0(v_0)
\end{align}
for some $0<b\leq\lvert\cY\rvert-1$. Then the first element $\hat{f}^{\mathsf{t}}(a,0)$ of \eqref{eq1:proof_prop:type} must be $v_0$.
Next, for each $0<b\leq\lvert\cY\rvert-1$, find $v$ (if exists) such that
\begin{align}
 nQ_b(v)>nQ_0(v).
\end{align}
Then the $(b+1)$th element $\hat{f}^{\mathsf{t}}(a,b)$ of \eqref{eq1:proof_prop:type} must be $v$. If such a $v$ does not exist, $\hat{f}^{\mathsf{t}}(a,b)$ must be $v_0$. In this manner, the list \eqref{eq1:proof_prop:type} is uniquely determined from $(Q_b:b\in\cY)$.
Since $[a]_{\hat{f}^{\mathsf{t}}}$ is uniquely determined from \eqref{eq1:proof_prop:type}, Condition (\ref{c1:def:informative}) holds.

Next we verify Condition (\ref{c2:def:informative}).
Let 
\begin{align}
 \cC^*:=\{x\in\cX:f(x,\cdot)\text{ is constant}\}.
\end{align}
Note that $\cC^*$ may be empty but if $\cC^*\neq\emptyset$ then $\cC^*\in\bcX_{\hat{f}^{\mathsf{t}}}$.
Now, fix $(\bx,\by)\in\cX^n\times\cY^n$ and $\sigma$ satisfying $[\sigma(\bx)]_{\hat{f}^{\mathsf{t}}}=[\bx]_{\hat{f}^{\mathsf{t}}}$ arbitrarily, and let $Q:= f_n^{\mathsf{t}}(\bx,\by)$ and $Q':= f_n^{\mathsf{t}}(\sigma(\bx),\by)$.
Further, let $\cC_i:= [x_i]_{\hat{f}^{\mathsf{t}}}=[x_{\sigma(i)}]_{\hat{f}^{\mathsf{t}}}$ for all $i\in[1:n]$.
Then, for each $v\in\cV$,
\begin{align}
 nQ(v) &=\sum_{\substack{i\in[1:n]:\\ \cC_i\neq\cC^*}}\bm{1}[f(x_i,y_i)=v]+\sum_{\substack{i\in[1:n]:\\ \cC_i=\cC^*}}\bm{1}[f(x_i,y_i)=v]\\
&=\sum_{\substack{i\in[1:n]:\\ \cC_i\neq\cC^*}}\bm{1}[f(x_{\sigma(i)},y_i)=v]+\sum_{\substack{i\in[1:n]:\\ \cC_i=\cC^*}}\bm{1}[f(x_{\sigma(i)},y_i)=v]\\
 &=nQ'(v),
\end{align}
where the second equality holds by the following reasons: the first terms coincide since $f(x_{\sigma(i)},\cdot)=f(x_i,\cdot)$ whenever $\cC_i\neq\cC^*$; 
the second terms coincide since $f(x_i,\cdot)$ is constant if $\cC_i=\cC^*$.
\end{IEEEproof}

\medskip
\begin{IEEEproof}
[Proof of Proposition \ref{prop:modsum}]
We first show Condition (\ref{c1:def:informative}) of Definition \ref{def:informative} is satisfied.
Fix $a\in\cX$, $(\bx,\by)\in\cX^n\times\cY^n$, and $i\in[1:n]$ arbitrarily, and let
\begin{align}
 v_b:= f_n^{\oplus}(a\bx^{(-i)},b\by^{(-i)}),\quad b\in\cY.
\end{align}
Then we have
\begin{align}
 \hat{f}^{\oplus}(a,b)&=v_{b+1}-v_b\quad(\bmod\ m),\quad b\in\cY
\end{align}
where $v_{\lvert\cY\rvert}=v_0$.
In other words, the list
\begin{align}
 \left(\hat{f}^{\oplus}(a,b):b\in\cY\right)
\label{eq1:proof_prop:modsum}
\end{align}
can be uniquely determined from $(v_b:b\in\cY)$.
Since $[a]_{\hat{f}^{\oplus}}$ is uniquely determined from \eqref{eq1:proof_prop:modsum}, the condition (\ref{c1:def:informative}) holds.

Next we verify Condition (\ref{c2:def:informative}).
Fix $(\bx,\by)\in\cX^n\times\cY^n$ and $\sigma$ satisfying $[\sigma(\bx)]_{\hat{f}^{\mathsf{t}}}=[\bx]_{\hat{f}^{\mathsf{t}}}$ arbitrarily, and let
$\cC_i:= [x_i]_{\hat{f}^{\oplus}}=[x_{\sigma(i)}]_{\hat{f}^{\oplus}}$ for all $i\in[1:n]$.
Then, we have
\begin{align}
 f_n(\bx,\by)&=\sum_{\cC\in\bcX_{\hat{f}^\oplus}}\sum_{\substack{i\in[1:n]:\\ \cC_i=\cC}} f(x_i,y_i)\quad(\bmod\ m)\\
&=\sum_{\cC\in\bcX_{\hat{f}^\oplus}}\sum_{\substack{i\in[1:n]:\\ \cC_i=\cC}} \left[f(x_i,0)+f(x_i,y_i)-f(x_i,0)\right]\quad(\bmod\ m)\\
&=\sum_{i=1}^n f(x_i,0)+\sum_{\cC\in\bcX_{\hat{f}^\oplus}}\sum_{\substack{i\in[1:n]:\\ \cC_i=\cC}} \left[f(x_i,y_i)-f(x_i,0)\right]\quad(\bmod\ m)\\
&\stackrel{\text{(a)}}{=}\sum_{i=1}^n f(x_{\sigma(i)},0)+\sum_{\cC\in\bcX_{\hat{f}^\oplus}}\sum_{\substack{i\in[1:n]:\\ \cC_i=\cC}} \left[f(x_{\sigma(i)},y_i)-f(x_{\sigma(i)},0)\right]\quad(\bmod\ m)\\
&=\sum_{\cC\in\bcX_{\hat{f}^\oplus}}\sum_{\substack{i\in[1:n]:\\ \cC_i=\cC}} f(x_{\sigma(i)},y_i)\quad(\bmod\ m)\\
&=f_n(\sigma(\bx),\by)
\end{align}
where the equality (a) holds by the following reasons: it is apparent that the first terms coincide since $\sigma$ is a permutation; 
the second terms coincide since, for any $i\in[1:n]$, $\sigma$ satisfies $[x_i]_{\hat{f}^{\oplus}}=[x_{\sigma(i)}]_{\hat{f}^{\oplus}}$ and thus
\begin{align}
 f(x_i,y_i)-f(x_i,0)=f(x_{\sigma(i)},y_i)-f(x_{\sigma(i)},0)\quad(\bmod\ m).
\end{align}
\end{IEEEproof}


\section{Proof of Lemmas \ref{lemma:compatible}, \ref{lemma:marginal_condition}, and \ref{lemma:connection-compatible-solvable}}\label{appendix:proof_lemmas}
\begin{IEEEproof}
[Proof of Lemma \ref{lemma:compatible}]
Since $(Q_0,\ldots,Q_{|{\cal Y}|-1}) \in {\cal Q}_n^{(i)}({\cal S},\bm{x},\bm{y})$, for any $b_1,b_2$ with $(x_i,b_1) \in {\cal S}$
and $(x_i,b_2) \in {\cal S}$, \eqref{eq:correct-sequence} implies 
\begin{align}
n Q_{b_1}(v) - n Q_{b_2}(v) = \mathbf{1}[ f(x_i,b_1) = v] - \mathbf{1}[ f(x_i,b_2) = v],~~~\forall v \in {\cal V}.
\end{align}
Thus, $x_i$ is compatible with $(Q_0,\ldots,Q_{|{\cal Y}|-1})$.
\end{IEEEproof}
\medskip

\begin{IEEEproof}
[Proof of Lemma \ref{lemma:marginal_condition}]
Suppose that $P_{\bar{X}\bar{Y}}^{(1)}, P_{\bar{X}\bar{Y}}^{(2)} \in {\cal P}_n({\cal A} \times {\cal B})$ be distinct joint types 
satisfying \eqref{eq:marginal-condition-x} and \eqref{eq:marginal-condition-y}.\footnote{If there is only one joint type satisfying
\eqref{eq:marginal-condition-x} and \eqref{eq:marginal-condition-y}, which is $P_{\bm{x}\bm{y}}$, there is nothing to be proved.}
Then, we shall show that \eqref{eq:unique-solution-condition} is satisfied.
First, we find a simple loop of ${\cal A} \times {\cal B}$ as follows. Let 
\begin{align}
\delta(x,y) := n P_{\bar{X}\bar{Y}}^{(1)}(x,y) - n P_{\bar{X}\bar{Y}}^{(2)}(x,y).
\end{align}
Then, \eqref{eq:marginal-condition-x} and \eqref{eq:marginal-condition-y} imply 
\begin{align}
\sum_{y \in {\cal B}} \delta(x,y) &= 0,~~~\forall x \in {\cal A}, \label{eq:delta-row-condition} \\
\sum_{x \in {\cal A}} \delta(x,y) &= 0,~~~\forall y \in {\cal B}. \label{eq:delta-column-condition}
\end{align}
Furthermore, we also have
\begin{align}
\sum_{x \in {\cal A} \atop y \in {\cal B}} \delta(x,y) = 0,
\end{align}
and
\begin{align}
\sum_{x \in {\cal A} \atop y \in {\cal B}} |\delta(x,y)|
 \label{eq:absolute-difference}
\end{align}
is strictly positive. We first pick any $(a_0,b_0) \in ({\cal A} \times {\cal B}) \cap {\cal S}$ such that $\delta(a_0,b_0) > 0$.
Then, we pick any $b_1 \in {\cal B}$ such that $b_1 \neq b_0$ and $\delta(a_0,b_1) < 0$, which must exist by \eqref{eq:delta-row-condition}.
Next, we pick any $a_1 \in {\cal A}$ such that $a_1 \neq a_0$ and $\delta(a_1,b_1) > 0$, which must exist by \eqref{eq:delta-column-condition}.
We continue this procedure by picking $b_i \in {\cal B}$ such that $b_i \neq b_j$ for $0 \le j < i$ and $\delta(a_{i-1},b_i) < 0$; or
by picking $a_i \in {\cal A}$ such that $a_i \neq a_j$ for $0 \le j < i$ and $\delta(a_i,b_i) > 0$.
We terminate the procedure when the only candidate is $b_0$ or $a_0$. If the procedure terminates by finding $b_0$
after picking $a_{m-1}$, then
\begin{align} 
\{ (a_0,b_0), (a_0,b_1),(a_1,b_1),(a_1,b_2),\ldots, (a_{m-2},b_{m-1}), (a_{m-1},b_{m-1}), (a_{m-1},b_0)\} 
\end{align}
is the desired simple loop; if the procedure terminates by finding $a_0$
after picking $b_{m-1}$, then
\begin{align}
\{ (a_0,b_{m-1}) ,(a_0,b_1),(a_1,b_1),(a_1,b_2),\ldots, (a_{m-2},b_{m-1}) \}
\end{align}
is the desired simple loop. Suppose that the former case occurred; the case with the latter 
proceed similarly with appropriate relabeling. 

Along the simple loop we found above, we modify $P_{\bar{X}\bar{Y}}^{(1)}$ as follows:
\begin{align}
P_{\bar{X}\bar{Y}}^{(1)}(a_i,b_i) &\to P_{\bar{X}\bar{Y}}^{(1)}(a_i,b_i) - \frac{1}{n}, \label{eq:modification-plus} \\
P_{\bar{X}\bar{Y}}^{(1)}(a_i,b_{i + 1~\mathrm{mod }~m}) &\to P_{\bar{X}\bar{Y}}^{(1)}(a_i,b_{i +1~\mathrm{mod}~m}) + \frac{1}{n}, \label{eq:modification-minus}
\end{align} 
and other components remain unchanged.
Since ${\cal A} \times {\cal B}$ is solvable, the simple loop must satisfy 
\eqref{eq:balanced-condition}. Thus, 
\begin{align} 
\sum_{(x,y) \in {\cal A}\times {\cal B}: \atop f(x,y)=v} P_{\bar{X}\bar{Y}}^{(1)}(x,y)
 \label{eq:invariance}
\end{align}
remain unchanged for every $v \in {\cal V}$ by the above modification procedure, \eqref{eq:modification-plus} and \eqref{eq:modification-minus}.
On the other hand, \eqref{eq:absolute-difference} strictly decrements by the above modification procedure. 

For the modified $P_{\bar{X}\bar{Y}}^{(1)}$ and the corresponding $\delta(x,y)$, we look for a simple loop again in the same manner as above,
and modify $P_{\bar{X}\bar{Y}}^{(1)}$ along the found simple loop. We continue this process until \eqref{eq:absolute-difference} become $0$,
which implies $P_{\bar{X}\bar{Y}}^{(1)}$ coincides with $P_{\bar{X}\bar{Y}}^{(2)}$. Since \eqref{eq:invariance} remain unchanged by the above 
modification procedure, \eqref{eq:unique-solution-condition} must have been satisfied in the first place. 
\end{IEEEproof}

\begin{IEEEproof}
[Proof of Lemma \ref{lemma:connection-compatible-solvable}]
The latter statement follows from the former statement since ${\cal E}({\cal S},f)$ is the set of all maximal solvable hyperedges.
Thus, we prove the former statement. Suppose that $e(Q_0,\ldots,Q_{|{\cal Y}|-1})$ is not solvable. Then, there exists a simple
loop
\begin{align}
\{ (a_0,b_0),(a_0,b_1),\ldots,(a_{m-1},b_{m-1}), (a_{m-1},b_0) \} \subseteq (e(Q_0,\ldots,Q_{|{\cal Y}|-1}) \times {\cal Y}) \cap {\cal S}
\end{align}
that violates \eqref{eq:balanced-condition} for some $v^* \in {\cal V}$. Since $a_0,\ldots,a_{m-1}$ are compatible with $(Q_0,\ldots,Q_{|{\cal Y}|-1})$,
we have
\begin{align}
n Q_{b_0}(v^*) - n Q_{b_1}(v^*) &= \mathbf{1}[ f(a_0,b_0) = v^* ] - \mathbf{1}[ f(a_0,b_1) = v^* ], \label{eq:loop-compatible-0} \\
n Q_{b_1}(v^*) - n Q_{b_2}(v^*) &= \mathbf{1}[ f(a_1,b_1) = v^* ] - \mathbf{1}[ f(a_1,b_2) = v^* ], \label{eq:loop-compatible-1} \\
&~ \vdots \nonumber \\
nQ_{b_{m-1}}(v^*) - n Q_{b_0}(v^*) &= \mathbf{1}[ f(a_{m-1},b_{m-1}) = v^* ] - \mathbf{1}[ f(a_{m-1},b_0) = v^*]. \label{eq:loop-compatible-m-1}
\end{align}
We find that the summation of the left hand sides of \eqref{eq:loop-compatible-0}-\eqref{eq:loop-compatible-m-1} is $0$; on the other hand,
since \eqref{eq:balanced-condition} is violated for $v^*$, the summation of the right hand sides of \eqref{eq:loop-compatible-0}-\eqref{eq:loop-compatible-m-1} is not $0$,
which is a contradiction. Thus, $e(Q_0,\ldots,Q_{|{\cal Y}|-1})$ is solvable.
\end{IEEEproof}

\subsection*{Acknowledgements} 

The authors would like to thank anonymous reviewers for valuable comments,
which improved the presentation of the paper. 


\bibliographystyle{../../IEEETranS}
\bibliography{../../reference}
\end{document}